\documentclass[a4paper,UKenglish,cleveref, autoref, thm-restate,authorcolumns, colorlinks]{lipics-v2021}
\usepackage{graphicx} 
\usepackage{amsthm}
\usepackage{amsmath}
\usepackage{amsfonts}
\usepackage{mathtools}
\usepackage{xcolor}
\usepackage{cleveref}
\usepackage{url}
\usepackage{todonotes}
\usepackage{etoolbox}
\usepackage[bibliography=common]{apxproof}

\usepackage{tikz}
\usetikzlibrary{automata, positioning, arrows,calc}

\newcommand\dom{\mathit{dom}}

\newcommand\calD{\mathcal{D}}
\newcommand\calF{\mathcal{F}}
\newcommand\calS{\mathcal{S}}
\newcommand\calT{\mathcal{T}}
\newcommand\calU{\mathcal{U}}
\newcommand\delay{\mathit{delay}}

\newcommand{\trans}{\mathcal{T}}
\newcommand{\transs}{\mathcal{S}}
\newcommand\dia[2]{\ensuremath{\mathit{dia}_{#2}(#1)}}

\newcommand{\dlev}{\ensuremath{d_l}}
\newcommand{\TP}[0]{\textbf{TP}}
\newcommand{\ATP}[0]{\textbf{ATP}}
\newcommand{\STP}[0]{\textbf{STP}}
\newcommand{\HTP}[0]{\textbf{HTP}}
\newcommand{\dHam}[2]{d_{h}(#1,#2)}

\appendixprelim{}

\newrobustcmd{\manu}[2][]{{\color{black}\todo[color=blue!30,noshadow,#1]{{\bf Manu:} #2}}\ignorespaces}
\newrobustcmd{\isma}[2][]{{\color{black}\todo[color=green!30,noshadow,#1]{{\bf Isma:} #2}}\ignorespaces}
\newrobustcmd{\saina}[2][]{{\color{black}\todo[color=pink!30,noshadow,#1]{{\bf Saina:} #2}}\ignorespaces}
\newrobustcmd{\khush}[2][]{{\color{black}\todo[color=red!45,noshadow,#1]{{\bf Khush:} #2}}\ignorespaces}

\title{Approximate Problems for Finite Transducers}

\author{Emmanuel Filiot}{Université libre de Bruxelles}{efiliot@ulb.e}{https://orcid.org/0000-0002-2520-5630}{He is a senior research associate at the National Fund for Scientific Research (F.R.S.-FNRS) in Belgium. His work is supported by the FNRS project T011724F.}
\author{Isma\"el Jecker}{Université de Franche-Comté}{ismael.jecker@gmail.com}{https://orcid.org/0000-0002-6527-4470}{}
\author{Khushraj Madnani}{Max Planck Institute for Software Systems}{kmadnani@mpi-sws.org}{https://orcid.org/0000-0003-0629-3847}{}
\author{Saina Sunny}{Indian Institute of Technology Goa}{saina19231102@iitgoa.ac.in}{https://orcid.org/0009-0005-1366-0168}{}

\authorrunning{E. Filiot, I. Jecker, K. Madnani and S. Sunny} 

\ccsdesc{Theory of computation~Transducers}

\ccsdesc[500]{Theory of computation~Quantitative automata}

\hideLIPIcs  

\nolinenumbers 

\keywords{Finite state transducers, Edit distance, Determinisation, Functionality}

\begin{document}

\maketitle

\begin{abstract}
    Finite (word) state transducers extend finite state automata by defining a binary relation over finite words, called \emph{rational  relation}. If the rational relation is the graph of a function, this  function is said to be rational. The class of sequential functions is a strict subclass of rational functions, defined as the functions
    recognised by input-deterministic finite state
    transducers. The class membership problems between those classes
    are known to be decidable. We consider approximate versions of
    these problems and show they are decidable as well. This includes
    the \emph{approximate functionality problem}, which asks whether
    given a rational relation (by a transducer), is it \emph{close} to a
    rational function, and the \emph{approximate determinisation problem}, which
    asks whether a given rational function is close to a sequential
    function. We prove decidability results for several classical
    distances, including Hamming and Levenshtein edit distance. Finally, we
    investigate the \emph{approximate uniformisation problem}, which
    asks, given a rational relation $R$, whether there exists a sequential
    function that is close to some function uniformising $R$. As its
    exact version, we prove that this problem is undecidable.
    
\end{abstract}

\section{Introduction}

Finite (state) transducers are a fundamental automata model to compute functions from words to words. The literature on finite state transducers is rich, and dates back to the early days of computer science, where they were called generalised sequential machines~\cite{Raney1958SequentialF,Ginsburg1968-GINACO-2}. See also~\cite{DBLP:conf/stacs/MuschollP19,DBLP:journals/siglog/FiliotR16} and the references therein. Finite state transducers extend finite automata with outputs on their transitions, allowing them to produce none or several symbols. While finite automata recognise languages of (finite) words, finite transducers recognise binary relations from words to words, called \emph{rational relations}. When the rational relation is the graph of a function, it is called a \emph{rational function}. This subclass is decidable within the class of rational relations. In particular, given a finite transducer $T$, it is decidable in \textsc{PTime} whether $T$ recognises a function~\cite{GurariI83,twinningproperty}. In that case, $T$ is said to be \emph{functional}. Beyond the fact that it is a natural restriction, the class of functional transducers is of high importance, as many problems known to be undecidable for transducers (such as inclusion and equivalence), become decidable under the functional restriction.

\subparagraph*{Determinisation} It turns out that non-determinism is needed for finite transducers to capture rational functions. A canonical example is the function $f_{\sf last} : \{a,b\}^*\rightarrow \{a,b\}^*$ that moves the last symbol upfront. For example, $f_{\sf last}(ab\underline{b}) = \underline{b}ab$ and $f_{\sf last}(ab\underline{a}) = \underline{a}ab$. Since the number of symbols that have to be read before reading the last symbol is arbitrarily large, a finite transducer recognising $f_{\sf last}$ needs non-determinism to guess the last symbol, as illustrated by the following transducer:
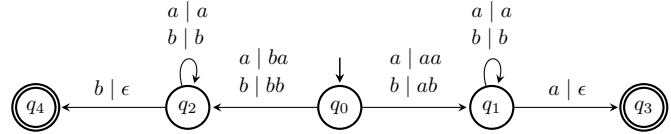
\begin{center}
\begin{tikzpicture}[draw=black,shorten >=1pt, auto, node distance=2cm, on grid, >=stealth,scale=0.8, every node/.style={scale=0.8}]
    \tikzset{state/.style={circle, draw=black, thick, minimum size=0.5cm}}
        \node[state, draw,initial,initial above, initial text=] (q0) {$q_0$};
        \node[state] (q1) [right=of q0] {$q_1$};
        \node[state] (q2) [left=of q0] {$q_2$};
        \node[state, accepting,draw] (q3) [right=of q1] {$q_3$};
        \node[state, accepting,draw] (q4) [left=of q2] {$q_4$};
    
        \path[->]
            (q0) edge node {$\begin{array}{lll}a\mid aa\\ b\mid ab\end{array}$} (q1)
            (q0) edge [above]node {$\begin{array}{l}a\mid ba\\ b\mid bb\end{array}$} (q2)
            (q1) edge [loop above] node {$\begin{array}{l}a\mid a\\b\mid b\end{array}$} (q1)
            (q2) edge [loop above] node {$\begin{array}{l}a\mid a\\b\mid b\end{array}$} (q2)
            (q2) edge [above] node {$b\mid \epsilon$} (q4)
            (q1) edge node {$a\mid\epsilon$} (q3);
    \end{tikzpicture}
\end{center}
So non-determinism, unlike finite automata, brings some extra expressive power when it comes to finite transducers. On the other hand, non-determinism also yields some inefficiency issues when the input is received sequentially as a stream, because the whole input may have to be stored in memory until the first output symbol can be produced. This motivates the class of \emph{sequential functions}, as the rational functions recognised by \emph{input-deterministic} finite transducers, and the \emph{determinisation problem}: given an arbitrary finite transducer, does it recognise a sequential function? In other words, can it be (input-) determinised? This well-studied problem is known to be decidable in \textsc{PTime}~\cite{twinningproperty}. The determinisation problem is a central problem in automata theory. It has been for instance extensively considered for weighted automata~\cite{DBLP:journals/coling/Mohri97}, and a long-standing open problem is whether this problem is decidable for $(\mathbb{N},max,+)$-automata~\cite{DBLP:conf/stacs/KirstenL09}.

\subparagraph*{Approximate determinisation} The function $f_{\sf last}$ is not sequential, in other words, the latter transducer is not determinisable. It turns out that $f_{\sf last}$ is \emph{almost} sequential, in the sense that it is \emph{close to} some sequential function, for instance the identity function ${\sf id}$. "Close to" can be defined in different ways, by lifting standard distances between words to functions and relations. Two classical examples are
the Hamming distance and the Levenshtein distance, which respectively measure the minimal number of letter substitutions (respectively letter substitutions, insertion and deletion) to rewrite a word into another. A distance $d$ between words is lifted to functions $f_1,f_2$ with the same input domain, by taking the supremum of $d(f_1(u),f_2(u))$ for all words $u$ in their domain. Coming back to our example, $f_{\sf last}$ and ${\sf id}$ are close for the edit distance, in the sense that $d(f_{\sf last},{\sf id})$ is finite for $d$ the edit distance, but they are not close for the Hamming distance. This raises a natural and fundamental problem, called \emph{approximate determinisation} problem (for a distance $d$): given a finite transducer recognising a function $f$, does there exists a sequential function $s$ such that $d(f,s)$ is finite? The approximate determinisation problem has been extensively studied for weighted automata~\cite{DBLP:conf/lics/AminofKL11,DBLP:journals/algorithmica/BuchsbaumGW01,DBLP:journals/tcs/AminofKL13} and quantitative automata~\cite{DBLP:conf/fsttcs/BokerH12,DBLP:journals/corr/BokerH14}, but, to the best of our knowledge, nothing was known for transducers. However, if both $f$ and $s$ are given (by finite transducers), checking whether they are close (for various and classical edit distances) is known to be decidable, even if $s$ is rational but not sequential~\cite{editdistance}. This can be seen as the verification variant of approximate determinisation, while approximate determinisation is rather a synthesis problem, for which only $f$ is given, and which asks to generate $s$ if it exists.

\subparagraph*{Contributions} In this paper, our main result is the decidability of approximate determinisation of finite transducers, for a family $\mathcal{E}$
of edit distances, which include Hamming and Levenshtein distances. 
For exact determinisation, determinisable finite transducers are characterised by the so called twinning property (\TP)~\cite{twinningproperty,twinningproperty0, ChoffrutG99}, a pattern that requires that the delay between any two outputs on the same input must not increase when taking synchronised cycles of the transducer. 
As noticed in~\cite{editdistance}, bounded (Levenshtein) edit distance is closely related to the notion of word conjugacy. In this paper, we consider an approximate version of the twinning property (\ATP{}), with no constraints on the delay, but instead requires that the output words produced on the synchronised loops are conjugate. It turns out that \ATP{} is not sufficient to characterise approximately determinisable transducers, and an extra property is needed, the strongly connected twinning property (\STP{}), which requires that the \TP{} holds within SCCs of the finite transducer. We show that a transducer $\calT$ is approximately determinisable (for Levenshtein distance) iff both \ATP{} and \STP{} hold and, if they do, we show how to approximately determinise $\calT$. We also prove that $\ATP{}$ and $\STP{}$ are both decidable.

For Hamming distance, which only allows letter substitutions, determinizable transducers are characterized by both \STP{} and another property called Hamming twinning property (\HTP{}), which roughly requires that outputs on synchronised cycles have same length and do not mismatch. We also show that \HTP{} is decidable, which entails decidability of approximate determinisation for Hamming distance.

We also consider another fundamental problem: the \emph{approximate functionality problem}. Informally, the approximate functionality problem asks whether a given rational relation $R$ is \emph{almost} a rational function, in the sense that $d(R,f)$ is finite for some rational function $f$, where $d(R,f)$ is now the supremum, for all $(u,v)\in R$, of $d(v,f(u))$. We prove that the approximate functionality problem is decidable for all the distances in $\mathcal{E}$. We prove this problem to be decidable for classical distances, including Hamming and Levenshtein. 

Finally, we consider the \emph{approximate (sequential) uniformisation problem}. In its exact version, this problem asks whether for a given rational relation $R$, there exists a sequential function $f$ with the same domain, and whose graph is included in $R$. This problem is closely related to a synthesis problem, but is unfortunately undecidable~\cite{LoC,DBLP:conf/icalp/FiliotJLW16}. We consider its approximate variant, where instead of requiring that the graph of $f$ is included in $R$, we require that it is close to some function whose graph is included in $R$. However, despite this relaxation, we show that the problem is still undecidable.

\subparagraph*{Other related works} Variants of the determinisation problem have been considered in the literature~\cite{DBLP:conf/icalp/FiliotJLW16}. However, this work considers the \emph{exact} determinisation of a transducer $\calT$, by some sequential transducer which is \emph{close} to $\calT$, for some notions of structural similarity between transducers.  Robustness and continuity notions for finite transducers have been introduced in~\cite{DBLP:conf/fsttcs/HenzingerOS14}. While those notions are also based on word distances, the problems considered are different from ours.

\section{Preliminaries}

For every $k \in \mathbb{N}$, we let $[k]$ denote the set $\{1,\ldots, k\}$. 

\subparagraph*{Words} Let $A$ or $B$ denote finite alphabets of letters. A word is a sequence of letters. The empty word is denoted by $\epsilon$. The length of a word is denoted by $|w|$, in particular $|\epsilon|=0$. The $i$th letter of a word $w$, for $i\in\{1,\dots,|w|\}$, is denoted by $w[i]$. The primitive root of $w$ is the shortest word $\rho_w$ such that $w\in\rho_w^*$. 
The set of all finite words over the alphabet $A$ is denoted by $A^*$. A relation $R\subseteq A^*\times B^*$ is sometimes called a \emph{transduction}, and is said to be \emph{functional} if it is the graph of a function. We let $\dom(R) = \{ u\in A^*\mid \exists v\in B^*, (u,v)\in R\}$ be the \emph{domain} of $R$, and for all $u\in A^*$, we let $R(u) = \{v\in B^*\mid (u,v)\in R\}$. Note that $u\in\dom(R)$ iff $R(u)\neq\varnothing$.

\subparagraph*{Finite State Automata and Transducers} A (non-deterministic) finite automaton over an alphabet $A$ is denoted as a tuple ${\cal A} = (Q, I, \Delta, F)$ where $Q$ is the finite set of states, $I\subseteq Q$ the set of initial states, $F\subseteq Q$ the set of final states, and $\Delta\subseteq Q\times A \times Q$ the transition relation. A \emph{run} on a word $\sigma_1\dots \sigma_n$ is a sequence of states $q_1\dots q_{n+1}$ such that  $(q_i,\sigma_i,q_{i+1})$ for all $i\in[n]$. It is accepting if and only if $q_1\in I$ and $q_{n+1}\in F$. We denote by $L(\cal A)$ the language accepted by $\cal A$, i.e. the set of words on which there is an accepting run. An automaton is said to be \emph{trim} if for any state $q$, there exists at least one accepting run visiting $q$. When ${\cal A}$ is deterministic, we denote the transition function as $\delta : Q\times A \rightarrow Q$.

A \emph{rational transducer} $\cal T$ over an input alphabet $A$ and an output alphabet $B$ is a (non-deterministic) finite automaton over a finite subset of $A^*\times B^*$. A transition $(p,(u,v),q)$ of $\cal T$ is often denoted $p\xrightarrow{u\mid v}_{\cal T} q$. More generally, we write $p\xrightarrow{u_1\dots u_n\mid v_1\dots v_n}_{\calT} q$ whenever there exists a run $\rho$ of $\calT$ from $p$ to $q$ on $(u_1,v_1)\dots (u_n,v_n)$. The \emph{relation recognised} by $\cal T$, denoted as $R_{\cal T}\subseteq A^*\times B^*$, is defined as 
$R_{\cal T} = \{ (u,v)\mid\exists q_0\in I,q_f\in F,q_0\xrightarrow{u\mid v}_{\calT} q_f\}$. The relations recognised by rational transducers are called \emph{rational relations}.

When rational transducers recognise functions, it is often convenient to restrict their transitions to input words of length one exactly, modulo the possibility of producing a word on accepting states. This defines the so-called class of real-time transducers, which is expressively equivalent to rational transducers when restricted to functions. 
Formally, a \emph{real-time transducer} (or simply, a transducer in the sequel) $\calT$ over input alphabet $A$ and output alphabet $B$ is given by a tuple ${\cal T} = (Q, I,\Delta, F, \lambda)$ where 
$(Q,I,\Delta, F)$ is a finite automaton over a finite subset of $A\times B^*$, and $\lambda : F\rightarrow B^*$ is a final output function.

Given a word $u = a_1 \cdots a_n \in A^*$ where $a_i\in A$ for all $i$, a run $\rho$ of $\calT$ over $u$ is a sequence $q_0 \cdots q_n$ such that $q_0 \in I$ and $(q_{i-1},a_i,v_i,q_i) \in \Delta$ for all $i \in [n]$. The \emph{input word} of the run $\rho$ is $u = a_1 \cdots a_n$ and the \emph{output word} of $\rho$ is $v_1 \cdots v_n \cdot \lambda(q_n)$ if $q_n$ is a final state; otherwise, $v_1 \cdots v_n$. As for rational transducers, the relation $R_{\cal T}$ recognised by $\cal T$ is defined as the set of pairs $(u,v)$ such that $u$ (resp. $v$) is the input (resp. output) word of some accepting run. We often confuse $\cal T$ with $R_{\cal T}$, and may write $\dom(\calT)$ for $\dom(R_\calT)$, or $\calT(u)$ for $R_\calT(u)$.

 The \emph{underlying automaton} of $\calT$ is the automaton obtained by projection on inputs, i.e. the automaton ${\cal A} = (Q, I, \Delta', F)$ such that $\Delta' = \{(q,a,q')\mid \exists (q,a,v,q')\in\Delta\}$. Note that $\dom(\calT) = L({\cal A})$. The transducer $\calT$ is said to be \emph{trim} if its underlying automaton is trim.

  The cartesian product, denoted $\calT_1 \times \calT_2$,
  of two transducers $\calT_i = (Q_i, I_i,\Delta_i, F_i, \lambda_i)$, $i\in[2]$, is the transducer $(Q_1 \times Q_2, I_1 \times I_2,\Delta, F_1 \times F_2, \lambda)$  where $((p_1,p_2),a,(v_1,v_2),(q_1,q_2)) \in \Delta$ if $(p_i,a,v_i,q_i) \in \Delta_i$ for $i \in [2]$, and, $\lambda(p_1,p_2) = (\lambda_1(p_1),\lambda_2(p_2))$ for $(p_1,p_2) \in F_1 \times F_2$.

\subparagraph*{(Sub)classes of rational functions} Let $\calT$ be a real-time transducer. When $R_\calT$ is functional, $\calT$ is said to be functional as well, and the functions recognised by functional transducers are called \emph{rational functions}. If the underlying automaton of $\calT$ is unambiguous (i.e., has at most one accepting run on any input), $\calT$ is referred to as an \emph{unambiguous transducer}. It is well-known that a function is recognised by real-time transducer iff, it is recognised by a rational transducer~\cite{Berstel} iff, it is recognised by an unambiguous transducer~\cite{eilenberg1974automata}.

\emph{Sequential transducers} are those whose underlying automaton is deterministic and they define functions known as \emph{sequential functions}. In that case, we denote the transition function as $\delta : Q\times A\rightarrow Q\times B^*$. The functions recognised by transducers that are
finite disjoint unions of sequential transducers are called \emph{multi-sequential functions}~\cite{DBLP:conf/stacs/ChoffrutS86,DBLP:journals/ijfcs/JeckerF18}. The symmetric counterpart of multi-sequential is the class of \emph{series-sequential functions}.

A transducer $\calT$ is \emph{series-sequential} if it is a finite disjoint union of sequential transducers $\calD_1, \ldots , \calD_k$
where additionally, for every $1 \leq i < k$,  there is
a single transition from a (not necessarily final) state $q_i$ of $\calD_i$
to the initial state of the next transducer $\calD_{i+1}$. Moreover, the initial state of $\calT$ is the initial state of $\calD_1$ (the initial states of $\calD_i$ is not considered as initial in $\calT$, for all $2\leq i\leq k$).
In particular, non-determinism is allowed, for all $1\leq i<k$, only in state $q_i$, from which it is possible to move to $\calD_{i+1}$ or to stay in $\calD_i$. We denote\footnote{This notation should not be confused with the split-sum operator of~\cite{DBLP:conf/csl/AlurFR14}, which is semantically different.} such a transducer as $\calD_1 \cdots \calD_k$.
A function is \emph{series-sequential} if it is recognised
by a series-sequential transducer.

\subparagraph{Distances between words and word functions} We recall that a metric on a set $E$ is a mapping
$d : E^2\rightarrow\mathbb{R}^+\cup \{\infty\}$ satisfying the separation, symmetry and triangle inequality axioms. 
Classical metrics between finite words are 
the \emph{edit distances}.
An edit distance between two words is the minimum number of edit operations required to rewrite a word to another if possible, and $\infty$ otherwise. Depending on the set of allowed edit operations, we get different edit distances.
\Cref{table:editdistance} gives widely used edit distances with their edit operations.

\begin{table}[h!]
\centering 
\begin{tabular}{|p{4.8cm}|p{1.2cm}|p{6cm}|}
\hline
Edit Distances & Notation & Edit Operations \\
\hline
Hamming   & \ $d_h$   & letter-to-letter substitutions   \\
Longest Common Subsequence   & \ $d_{lcs}$   & insertions and deletions   \\
Levenshtein   & \ $\dlev$   & insertions, deletions and substitutions   \\
Damerau-Levenshtein   & \ $d_{dl}$   & insertions, deletions, substitutions and swapping adjacent letters   \\
\hline
\end{tabular}
\caption{Edit Distances \cite{editdistance}}
\label{table:editdistance}
\end{table}

Distances between words can be lifted to that between functions from words to words.
\begin{definition}[Metric over Functions \cite{editdistance}]\label{def:distfunc}
 Let $d$ be a metric on words over some alphabet $B$. Given two partial functions $f_1,f_2: A^* \rightarrow B^*$, the distance between $f_1$ and $f_2$ is defined as 
$$d(f_1,f_2) = \begin{cases} \sup \left \{\,  d(f_1(w), f_2(w)) \,\mid\,  w \in \dom(f_1) \right \} & \text{ if $\dom(f_1) = \dom(f_2)$} \\
 \infty & \text{ otherwise } 
 \end{cases}$$
\end{definition}

It is shown that $d$ is a metric over functions (Proposition 3.2 of \cite{editdistance}). 
The distance between two functional transducers is defined as the distance between the functions they recognise. A notion closely related to the distance between functions is \emph{diameter} of a relation. The diameter of a relation $R$ w.r.t.~metric $d$, denoted by $\dia{R}{d}$ is defined to be the supremum of the distance of every pair in the relation, i.e.,
$\dia{R}{d} = \sup \{d(u,v) \mid (u,v) \in R\}$.

The distance between rational functions and diameter of rational relation w.r.t.~the metrics given in \Cref{table:editdistance} are computable \cite{editdistance}. The computability of distance and diameter relies on the notion of conjugacy.  Two words $u$ and $v$ are conjugate if there exist words $x,y$ such that $u=xy$ and $v =yx$. In other words, they are cyclic shifts of each other. For example, words $aabb$ and $bbaa$ are conjugate with $x=aa$ and $y=bb$. But $aabb$ and $abab$ are not conjugate. Conjugacy is an equivalence relation over words.

\begin{proposition}\label{prop:conj}
    Let $x,y,x',y',u,v$ be words and $c,C \in \mathbb{N}$. For any metric $d$ in \Cref{table:editdistance}, if $d(xu^ky,x'v^{ck}y') \leq C$ for all $k \geq 0$, then $|u| = |v^{c}|$ and the primitive roots of $u$ and $v$ are conjugate. 
\end{proposition}

\begin{proof}
    Since $d(xu^ky,x'v^{ck}y') \leq C$, we get $|u| = |v^c|$. Otherwise, as $k$ increases the length difference of $xu^{k}y$ and $x'v^{ck}y'$ increases, and hence their distance will not be bounded.

    Since $|u| = |v^c|$, either both $u$ and $v$ are nonempty, or $u=v= \epsilon$. In the latter case, $u$ and $v$ are conjugate. Assume that $u$ and $v$ are nonempty words. Take $k = 2^{|u| + |v|}$. Since $d(xu^ky,x'v^{ck}y') \leq C$ there exist large portions of $u$’s and $v$’s that match. In fact, $u$'s and $v$'s overlap at least of length $|u| + |v|$. By Fine and Wilf's\footnote{The Fine and Wilf's theorem states that if some powers of two words $u$ and $v$ share a common factor of length $|u| + |v| - \gcd(u,v)$ then their primitive roots are conjugate.} theorem, the primitive roots of $u$ and $v$ are conjugate. 
\end{proof}

\begin{proposition}[\cite{editdistance}]\label{prop:levboundconj}
    Given a rational relation $R$ defined by a transducer $\calT$, $\dia{R}{d} < \infty$ for $d \in \{\dlev,d_{lcs},d_{dl}\}$ if and only if every pair of input-output words generated by loops in $\calT$ are conjugate.
   
\end{proposition}

\section{Twinning Properties}
 The class of sequential functions has been characterised by transducers satisfying the so called \emph{twinning property}. In this section, we recall this property and introduce three variants --- approximate twinning property, strongly connected twinning property and Hamming twinning property. We also show that these properties on transducers are decidable as well as independent of the representation of the transducers.
All these properties are expressed as particular conditions on twinning patterns. A \emph{twinning pattern} for a transducer $\calT = (Q, I, \Delta, F, \lambda)$ over $A,B$ is a tuple $(p_1,q_1,p_2,q_2,u,v,u_1,v_1,u_2,v_2)\in Q^4\times (A^*)^2 \times (B^*)^4$ such that the following runs (graphically depicted) exist:
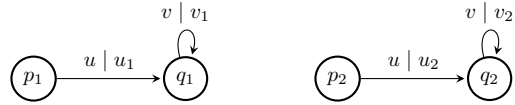
\begin{center}
\begin{tikzpicture}[draw=black,shorten >=1pt, auto, node distance=2cm, on grid, >=stealth,scale=0.8, every node/.style={scale=0.8}]
    \tikzset{state/.style={circle, draw=black, thick, minimum size=0.5cm}}
        \node[state] (q0) {$p_1$};
        \node[state] (q1) [right=of q0] {$q_1$};
        \node[state] (p0) [right=of q1]{$p_2$};
        \node[state] (p1) [right=of p0] {$q_2$};
    
        \path[->]
            (q0) edge [above] node {$u \mid u_1$} (q1)
            (q1) edge [loop above] node {$v \mid v_1$} (q1)
            (p0) edge [above] node {$u \mid u_2$} (p1)
            (p1) edge [loop above] node {$v \mid v_2$} (p1);
    \end{tikzpicture}
\end{center}

The longest common prefix of any two words is denoted by $u\wedge v$. The \emph{delay} between $u$ and $v$, denoted by $\delay(u,v)$, is the pair $(u',v')$ such that $u = (u\wedge v)u'$ and $v=(u\wedge v)v'$.

\begin{definition}[Twinning Property (\TP)]\label{def:tp}
    Let $\calT$ be a trim transducer.
    We say that $\calT$ satisfies \emph{twinning property} if for each twinning pattern
    $(p_1,q_1,p_2,q_2,u,v,u_1,v_1,u_2,v_2)$ such that $p_1,p_2$ are initial,   $\delay(u_1,u_2) = \delay(u_1v_1,u_2v_2)$ holds. 
\end{definition}

It is well-known that a function recognised by a transducer $\calT$ is sequential iff $\calT$ satisfies the twinning property~\cite{DBLP:journals/tcs/Choffrut77,twinningproperty}. We now define its approximate variant.

\begin{definition}[Approximate Twinning Property (\ATP)]\label{def:atp}
    A trim transducer $\calT$ satisfies \emph{approximate twinning property}
    if for each twinning pattern $(p_1,q_1,p_2,q_2,u,v,u_1,v_1,u_2,v_2)$
    where $p_1,p_2$ are initial, the words $v_1$ and $v_2$ are conjugate. 
\end{definition}

\begin{definition}[Strongly Connected Twinning Property (\STP)]\label{def:stp}
    A trim transducer $\calT$ satisfies \emph{strongly connected twinning property} if for each twinning pattern $(p_1,q_1,p_2,q_2,u,v,u_1,v_1,u_2,v_2)$ such that $p_1=p_2$ (not necessarily initial) and $p_1,q_1,q_2$ are in the same strongly connected component,  $\delay(u_1,u_2) = \delay(u_1v_1,u_2v_2)$ holds.
\end{definition}

\begin{definition}[Hamming Twinning Property (\HTP)]\label{def:htp}
    A trim transducer $\calT$ satisfies \emph{Hamming twinning property} if for each twinning pattern $(p_1,q_1,p_2,q_2,u,v,u_1,v_1,u_2,v_2)$ such that $p_1,p_2$ are initial, it holds that $|v_1|=|v_2|$ and there is no mismatch between $v_1$ and $v_2$, i.e. for all position $i\in[\text{max}(|u_1|,|u_2|),\text{min}(|u_1v_1|,|u_2v_2|)]$, $(u_1v_1)[i] = (u_2v_2)[i]$.
\end{definition}

\begin{proposition}\label{proposition:TPrelations}
    Each trim transducer satisfying \TP{} also satisfies \ATP{}, \STP{} and \HTP.
\end{proposition}

\begin{proof}
   Let $\calT$ be a trim transducer which satisfies \TP.

We first show it satisfies \ATP. Consider a twinning pattern $(p_1,q_1,p_2,q_2,u,v,u_1,v_1,u_2,v_2)$ such that $p_1,p_2$ are initial. Then, $\delay(u_1,u_2)=\delay(u_1v_1,u_2v_2) = \delay(u_1v_1^i,v_1v_2^i)$ for all $i\geq 0$, which implies that $|v_1|=|v_2|$. Therefore, for all $n\geq 0$, there exists $i$ such that $v_1^i$ and $v_2^i$ have a common factor of length at least $n$. By Fine and Wilf's theorem, it implies that $v_1$ and $v_2$ have conjugate primitive roots, and since $|v_1|=|v_2|$, we get that $v_1$ and $v_2$ are conjugate.

Now, consider \STP. Let $(p,q_1,p,q_2,u,v,u_1,v_1,u_2,v_2)$ be some twinning pattern where $p,q_1,q_2$ are in the same SCC. Since $\calT$ is trim, $p$ is accessible from the initial state $q_0$, by a run $q_0\xrightarrow{\alpha\mid \beta}_{\calT} p$. Note that $(q_0,q_1,q_0,q_2,\alpha u, v, \beta u_1,v_1,\beta u_2,v_2)$ is also a twinning pattern. Hence by the twinning property, $\delay(\beta u_1,\beta u_2) = \delay(\beta u_1v_1,\beta u_2v_2)$, from which we can conclude as $\delay(\beta u_1,\beta u_2) = \delay(u_1,u_2)$ and $\delay(\beta u_1v_1,\beta u_2v_2) = \delay(u_1v_1,u_2v_2)$.

Finally, if the \HTP{} is not satisfied, then considering a twinning pattern as in the definition of the \HTP, there is a position $i$ such that $(u_1v_1)[i] \neq (u_2v_2)[i]$, and therefore $\delay(u_1v_1,u_2v_2)\neq \delay(u_1v_1v_1,u_2v_2v_2)$, which fails the \TP, contradiction.
\end{proof}

The following lemma states that \ATP{}, \STP{} and \HTP{} is preserved between transducers up to finite edit distance, and so in particular between equivalent transducers. This shows that these properties do not depend on the representation of the transductions, not even on the representation of close transductions.

\begin{lemma}\label{lem:preserveTPs}
    Let $\trans$ and $\transs$ be two trim transducers
    satisfying $\dom(\trans) = \dom(\transs)$, and
    such that there exists an edit distance $d$ in \Cref{table:editdistance}
    and a constant $C \in \mathbb{N}$
    for which
    \[
    d(\trans(u),\transs(u))<C \text{ for all $u \in \dom(\trans)$}.
    \]
    Then for every ${\sf P}\in \{\ATP,\STP\}$,
    $\transs$ satisfies ${\sf P}$ if and only if $\trans$ satisfies ${\sf P}$.
    The statement also holds for $d = d_h$ and ${\sf P} = \HTP$.
\end{lemma}

\begin{proof}
    Let us suppose that $\transs$ satisfies ${\sf P}\in \{\ATP,\STP, \HTP\}$, and show that so does $\trans$. The converse is symmetric. 
    \subparagraph*{For \ATP{}}  Let $(p_1,p_2,q_1,q_2,u,v,u_1,u_2,v_1,v_2)$ be an instance of the twinning pattern
    for $\trans$ with $p_1$ and $q_1$ initial.
    Since  $\trans$ is trim,
    there exist words $w,w',w_1,w_2$ and two accepting states $p_f,q_f$ such that:
    $$
    \begin{array}{llllllllllll}
    p_1 & \xrightarrow{u|u_1}_\trans & p_2 & \xrightarrow{v|v_1}_\trans & p_2 &  \xrightarrow{w|w_1}_\trans & p_f \\
q_1 & \xrightarrow{u|u_2}_\trans & q_2 & \xrightarrow{v|v_2}_\trans & q_2 &  \xrightarrow{w'|w_2}_\trans & q_f \\    
    \end{array}
    $$
Since $\dom(\trans) = \dom(\transs)$,
for all $i\geq 1$, $uv^iw,uv^iw' \in \dom(\trans) = \dom(\transs)$.
Iterating the loop of $\trans$ on $v$ sufficiently many times
causes $\transs$ to also loop on some power of $v$.
Formally, there exist $a,b,c\in\mathbb{N}$,
states $p'_1,p'_2,q'_1,q'_2,p'_f,q'_f$
and words $u'_1,u'_2,v'_1,v'_2,w'_1,w'_2$ such that:

   $$
    \begin{array}{llllllllllll}
    p_1 & \xrightarrow{uv^{a}|u_1v_1^{a}}_\trans & p_2 & \xrightarrow{v^b|v_1^b}_\trans & p_2 &  \xrightarrow{v^cw|v_1^cw_1}_\trans & p_f \\
q_1 & \xrightarrow{uv^{a}|u_2v_2^{a}}_\trans & q_2 & \xrightarrow{v^b|v_2^b}_\trans & q_2 &  \xrightarrow{v^cw'|v_2^cw_2}_\trans & q_f \\    

    p'_1 & \xrightarrow{uv^{a}|u'_1}_\transs & p'_2 & \xrightarrow{v^b|v'_1}_\transs & p'_2 &  \xrightarrow{v^cw|w'_1}_\transs & p'_f \\
q'_1 & \xrightarrow{uv^{a}|u'_2}_\transs & q'_2 & \xrightarrow{v^b|v'_2}_\transs & q'_2 &  \xrightarrow{v^cw'|w'_2}_\transs & q'_f \\    
    \end{array}
    $$
Let $C \in \mathbb{N}$
such that $d(\trans(x),\transs(x)) \leq C$ for all $x \in \dom(\trans)$.
Then in particular for all $k\geq 0$: 
$$
d(u_1v_1^{a+kb+c}w_1,u_1'v_1'^kw_1')\leq C\text{ and }d(u_2v_2^{a+kb+c}w_2,u_2'v_2'^kw_2')\leq C.
$$
From the above inequalities, using \Cref{prop:conj},
we get that $\rho_{v_i}$ and $\rho_{v'_i}$, the primitive roots of $v_i$ and $v'_i$ respectively, are conjugate and $|v_i^{b}| = |v_i'|$ for $i \in [2]$.
Since $\transs$ satisfies the \ATP{},
we also have that $v'_1$ and $v'_2$ are conjugate, and hence their primitive roots are conjugate and $|v'_1| = |v'_2|$.
By transitivity of the conjugacy relation,
we get that $\rho_{v_1}$ and $\rho_{v_2}$ are conjugate. Further, $|v_1| = |v_2|$ since $|\rho_{v_1}| = |\rho_{v_2}|$ and $|v_1^b| = |v_1'| = |v_2'| =|v_2^b|$. Therefore, $v_1$ and $v_2$ are conjugate, and
thus the \ATP{} holds for $\trans$ as well.

\subparagraph*{For \STP{}} 
We suppose that $\transs$ satisfies the \STP{}, and show that $\trans$ satisfies it too.
Let us pick a twinning pattern $(p,q_1,p,q_2,u,v,u_1,v_1,u_2,v_2)$
in $\trans$ such that $p,q_1,q_2$ are in the same strongly connected component,
as in the definition of the \STP{} (Definition~\ref{def:stp}):
\[
    \begin{array}{lllllllllllll}
    \rho_1: & q_I& \xrightarrow{x|x_0}_\trans &
    p & \xrightarrow{u|u_1}_\trans &
    q_1 & \xrightarrow{v|v_1}_\trans &
    q_1 &  \xrightarrow{w^+|w_1}_\trans &
    p & \xrightarrow{z|z_0}_\trans &
    q_F, \\
    \rho_2: &
    q_I & \xrightarrow{x|x_0}_\trans &
    p & \xrightarrow{u|u_2}_\trans &
    q_2 & \xrightarrow{v|v_2}_\trans &
    q_2 &  \xrightarrow{w^-|w_2}_\trans &
    p & \xrightarrow{z|z_0}_\trans &
    q_F.    
    \end{array}
\]
Here the runs with inputs $x$ and $z$ are witnesses of the fact that $\trans$ is trim
and the runs with inputs $w^+$ and $w^-$ are witnesses 
of the fact that $p,q_1,q_2$ are in the same strongly connected component.
Similar to the proof of \ATP{},
we use these runs to build an instance of \STP{} in $\transs$.
To that end, let $N$ denote the number of states of $\transs$.
We study the run of $\transs$ on the word~:
\[
    y = xy_1y_2\ldots y_{2N}z, \textup{ where $y_i = (uv^{N^2}w^+)^iuv^{N^2}w^-$ for every $1 \leq i \leq 2N$}.
\]
Remark that $y \in \dom(\trans)$,
as a run over $y$ can be obtained by pumping and combining the subruns of $\rho_1$ and $\rho_2$.
Therefore, as $\dom(\trans) = \dom(\transs)$,
there is also a run of $\transs$ on $y$,
and since $\transs$ has $N$ states
this run will visit some state $r$ three times between the $y_i$.
Formally, there exist $1 \leq i < j < k \leq 2N$ such that,
if we let $\bar{x} = xy_1y_2 \ldots y_{i-1}$ and $\bar{z} = y_{k}y_{k+1}\ldots y_{2N}z$, we get:
\[
    \begin{array}{lllllllllllllll}
    r_I & \xrightarrow{\bar{x}|\bar{x}_0}_\transs &
    r & \xrightarrow{y_{i} \ldots y_{j-1}|y'}_\transs &
    r &  \xrightarrow{y_{j} \ldots y_{k-1}|y''}_\transs &
    r &  \xrightarrow{\bar{z}|\bar{z}_0}_\transs &
    r_F.    
    \end{array}
\]
Now remark that the input words labeling both loops on $r$
start with the prefix $(uv^{N^2}w^+)^iuv^{N^2}$,
which is followed by $w^-$ in the first loop and $w^+$ in the second.
Since $\transs$ has $N$ states, the last block of $v^{N^2}$ occurring in this common
prefix will visit a synchronised subloop in both loops.
Formally, there is a decomposition
$\bar{u}\bar{u}'\bar{v}\bar{w}^+$ of $y_{i} \ldots y_{j-1}$
and a decomposition $\bar{u}\bar{u}'\bar{v}\bar{w}^-$ of $y_{j} \ldots y_{k-1}$
satisfying:
\begin{gather*}
    \bar{u} = (uv^{N^2}w^+)^i, \quad \quad \quad 
    \bar{u}' = uv^a, \quad \quad \quad 
    \bar{v} = v^b, \quad \quad \quad 
    \bar{w}^+ = v^{c}w^-y_{i+1} \ldots y_{j-1}, \\
    \hfill \bar{w}^- = v^{c}w^+(uv^{N^2}w^+)^{j-i-1}uv^{N^2}w^-y_{j+1} \ldots y_{k-1},\hfill\\
    \begin{array}{lllllllllllll}
    \bar{\rho}_1: &
    s_I  \xrightarrow{\bar{x}|\bar{x}_0}_\transs 
    r  \xrightarrow{\bar{u}|\bar{u}_1}_\transs 
    r_1  \xrightarrow{\bar{u}'|\bar{u}_1'}_\transs 
    s_1  \xrightarrow{\bar{v}|\bar{v}_1}_\transs 
    s_1  \xrightarrow{\bar{w}^+|\bar{w}_1}_\transs 
    r  \xrightarrow{\bar{z}|\bar{z}_0}_\transs 
    s_F, \\
    \bar{\rho}_2: &
    s_I  \xrightarrow{\bar{x}|\bar{x}_0}_\transs 
    r  \xrightarrow{\bar{u}|\bar{u}_2}_\transs 
    r_2  \xrightarrow{\bar{u}'|\bar{u}_2'}_\transs 
    s_2  \xrightarrow{\bar{v}|\bar{v}_2}_\transs 
    s_2  \xrightarrow{\bar{w}^-|\bar{w}_2}_\transs 
    r \xrightarrow{\bar{z}|\bar{z}_0}_\transs 
    s_F.    
    \end{array}
\end{gather*}
Since $\transs$ satisfies the \STP{} by hypothesis,
$\delay(\bar{u}_1\bar{u}_1',\bar{u}_2\bar{u}_2') = \delay(\bar{u}_1\bar{u}_1'\bar{v}_1,\bar{u}_2\bar{u}_2'\bar{v}_2)$.
We now combine this with the fact that the distance
between the outputs of $\trans$ and $\transs$ is bounded by $C$
to show that $\delay(u_1,u_2) = \delay(u_1v_1,u_2v_2)$,
which concludes the proof of the lemma.
First, let us pump the subruns of $\rho_1$
and $\rho_2$ to match the inputs of $\bar{\rho}_1$ and $\bar{\rho}_2$~:
\begin{gather*}
    \begin{array}{lllllllllllll}
    \rho'_1: &
    q_I \xrightarrow{\bar{x}|x'_0}_\trans
    p \xrightarrow{\bar{u}|u_0}_\trans
    p \xrightarrow{\bar{u}'|u_1v_1^a}_\trans
    q_1 \xrightarrow{\bar{v}|v_1^b}_\trans
    q_1 \xrightarrow{\bar{w}^+|w'_1}_\trans
    p \xrightarrow{\bar{z}|z'_0}_\trans
    q_F, \\
    \rho'_2: &
    q_I \xrightarrow{\bar{x}|x'_0}_\trans
    p \xrightarrow{\bar{u}|u_0}_\trans
    p \xrightarrow{\bar{u}'|u_2v_2^a}_\trans
    q_2 \xrightarrow{\bar{v}|v_2^b}_\trans
    q_2 \xrightarrow{\bar{w}^-|w'_2}_\trans
    p \xrightarrow{\bar{z}|z'_0}_\trans
    q_F.    
    \end{array}
\end{gather*}
We then introduce a standard property of the delay,
proved for instance in~\cite{DBLP:journals/ijfcs/FiliotMR20}~:
\begin{claim}\label{claim:TP}
    For every words $s_1,s_2,t_1,t_2$ we have that $\delay(s_1,s_2) = \delay(s_1t_1,s_2t_2)$
    if and only if either $t_1 = t_2 = \epsilon$,
    or $|t_1| = |t_2|$ and there is no mismatch
    between $s_1t_1^{m}$ and $s_2t_2^{m}$ for all $m \in \mathbb{N}$.
\end{claim}
Therefore, in order to show that $\delay(u_1,u_2) = \delay(u_1v_1,u_2v_2)$,
we now suppose that $v_1 \neq \epsilon$ or $v_2 \neq \epsilon$,
and we show that $|v_1| = |v_2|$ and there is no mismatch between  $u_1v_1^{m}$ and $u_2v_2^{m}$
for all $m \in \mathbb{N}$.

\subparagraph*{Same output length}
We first show that $|v_1| = |v_2|$.
Observe that for every $\lambda \in \mathbb{N}$ we have~: 
\[
\begin{array}{lllllll}
\transs(\bar{x}\bar{u}uv^{a + \lambda \cdot b}\bar{w}^+\bar{z})
& = & \bar{x}_0\bar{u}_1\bar{u}_1'\bar{v}_1^{\lambda}\bar{w}_1\bar{z}_0,
& \ &
\transs(\bar{x}\bar{u}uv^{a + \lambda \cdot b}\bar{w}^-\bar{z})
& = & \bar{x}_0\bar{u}_2\bar{u}_2'\bar{v}_2^{\lambda}\bar{w}_2\bar{z}_0,\\
\trans(\bar{x}\bar{u}uv^{a + \lambda \cdot b}\bar{w}^+\bar{z})
& = & x_0'u_0u_1v_1^{a + \lambda \cdot b}w'_1z_0',
& \ &
\trans(\bar{x}\bar{u}uv^{a + \lambda \cdot b}\bar{w}^-\bar{z})
& = & x_0'u_0u_2v_2^{a + \lambda \cdot b}w'_2z'_0.
\end{array}
\]
By the supposition in the statement of the Lemma,
$d(\bar{x}_0\bar{u}_1\bar{u}_1'\bar{v}_1^{\lambda}\bar{w}_1\bar{z}_0,
x_0'u_0u_1v_1^{a + \lambda \cdot b}w'_1z_0') \leq C$
and
$d(\bar{x}_0\bar{u}_2\bar{u}_2'\bar{v}_2^{\lambda}\bar{w}_2\bar{z}_0,
x_0'u_0u_2v_2^{a + \lambda \cdot b}w'_2z'_0) \leq C$.
This implies that $|v_1|^b = \bar{v}_1$ and  $|v_2|^b = \bar{v}_2$,
as otherwise choosing $\lambda$ large enough yields outputs
with a length difference too large  to be fixed with atmost $C$ number of edits. Hence, it suffices to show $|\bar{v}_1| = |\bar{v}_2|$. This is implied by the fact $\delay(\bar{u}_1\bar{u}_1',\bar{u}_2\bar{u}_2') = \delay(\bar{u}_1\bar{u}_1'\bar{v}_1,\bar{u}_2\bar{u}_2'\bar{v}_2)$ along with the 
\Cref{claim:TP}.

\subparagraph*{No mismatch}
As we have shown that $|v_1| = |v_2|$, let us now pick $m \in \mathbb{N}$
and show that there is no mismatch
between $u_1v_1^m$ and  $u_2v_2^m$ to conclude the proof via \Cref{claim:TP}.
Remark that the words $u_1v_1^m$
and $u_2v_2^m$ appear as subwords of the
outputs of the runs $\rho_1'$ and $\rho_2'$.
We now derive three key equations on output words
(\Cref{equ:loopSize}, \Cref{equ:pre} and \Cref{equ:STP})
by pumping and assembling parts of $\rho_1'$, $\rho_2'$,
$\bar{\rho}_1$ and $\bar{\rho}_2$.
We then combine these equations to get the desired result.
We begin by introducing some notation. Let $M\in \mathbb{N}$ satisfying
\[
M > \frac{m \cdot |v_1| + C + \max(|\bar{x}_0|,|x_0'|)}{b \cdot |v_1|}.
\]
Moreover, we consider the following output words~:
\[  \begin{array}{lllllllllllllll}
    \alpha_1 & = & u_0u_1v_1^{a + Mb}, & &
    \alpha_2 & = & u_0u_2v_2^{a + Mb}, & &
    \beta_1 & = & w_1', &&
    \beta_2 & = & w_2',\\
    \bar{\alpha}_1 & = &\bar{u}_1\bar{u}_1'\bar{v_1}^M, & &
    \bar{\alpha}_2 & = & \bar{u}_2\bar{u}_2'\bar{v_2}^M, & &
    \bar{\beta}_1 & = & \bar{w}_1, &&
    \bar{\beta}_2 & = & \bar{w}_2.
    \end{array}
\]
First, remark that for every $\lambda \in \mathbb{N}$ we have~:
\[
\begin{array}{lllllll}
\transs(\bar{x}(\bar{u}\bar{u}'\bar{v}^M\bar{w}^+)^\lambda\bar{z})
& = & \bar{x}_0(\bar{\alpha}_1\bar{\beta}_1)^\lambda\bar{z}_0,
& \ &
\transs(\bar{x}(\bar{u}\bar{u}'\bar{v}^M\bar{w}^-)^\lambda\bar{z})
& = & \bar{x}_0(\bar{\alpha}_2\bar{\beta}_2)^\lambda\bar{z}_0,\\
\trans(\bar{x}(\bar{u}\bar{u}'\bar{v}^M\bar{w}^+)^\lambda\bar{z})
& = & x_0'(\alpha_1\beta_1)^\lambda z_0',
& \ &
\trans(\bar{x}(\bar{u}\bar{u}'\bar{v}^M\bar{w}^+)^\lambda\bar{z})
& = & x_0'(\alpha_2\beta_2)^\lambda z_0'.
\end{array}
\]
Since the edit distance between
the outputs of $\trans$ and $\transs$ is bounded by $C$, this yields~:
\begin{equation}\label{equ:loopSize}
    |\alpha_1 \beta_1| = |\bar{\alpha}_1 \bar{\beta}_1| \textup{ and } 
    |\alpha_2 \beta_2| = |\bar{\alpha}_2 \bar{\beta}_2|.
\end{equation}
Second, for every $\lambda \in \mathbb{N}$ we have~:
\[
\begin{array}{lllllll}
\transs(\bar{x}(\bar{u}\bar{u}'\bar{v}^M\bar{w}^+\bar{u}\bar{u}'\bar{v}^M\bar{w}^+\bar{u}\bar{u}'\bar{v}^M\bar{w}^-)^\lambda\bar{z})
& = & \bar{x}_0(\bar{\alpha}_1\bar{\beta}_1\bar{\alpha}_1\bar{\beta}_1\bar{\alpha}_2\bar{\beta}_2)^\lambda\bar{z}_0,\\
\trans(\bar{x}(\bar{u}\bar{u}'\bar{v}^M\bar{w}^+\bar{u}\bar{u}'\bar{v}^M\bar{w}^+\bar{u}\bar{u}'\bar{v}^M\bar{w}^-)^\lambda\bar{z})
& = & x_0'(\alpha_1\beta_1\alpha_1\beta_1\alpha_2\beta_2)^\lambda z_0'.
\end{array}
\]
If we set $\lambda = C$, 
the hypothesis in the statement of the Lemma yields~: 
\[
d(x_0'(\alpha_1\beta_1\alpha_1\beta_1\alpha_2\beta_2)^C z_0',\bar{x}_0(\bar{\alpha}_1\bar{\beta}_1\bar{\alpha}_1\bar{\beta}_1\bar{\alpha}_2\bar{\beta}_2)^C\bar{z}_0) < C.
\]
Observe that at least one copy of the subword $\alpha_1\beta_1\alpha_1\beta_1\alpha_2\beta_2$
is preserved (i.e., no internal deletion or insertion) while editing
$x_0'(\alpha_1\beta_1\alpha_1\beta_1\alpha_2\beta_2)^\lambda z_0'$ into $\bar{x}_0(\bar{\alpha}_1\bar{\beta}_1\bar{\alpha}_1\bar{\beta}_1\bar{\alpha}_2\bar{\beta}_2)^\lambda\bar{z}_0$
via $C-1$ edits.
As this copy of $\alpha_1\beta_1\alpha_1\beta_1\alpha_2\beta_2$ can be shifted by at most
$C - 1$ letters, and as it is originally shifted by
$\max(|x_0' |,|\bar{x}_0|)$ letters with respect to 
its matching copy of $\bar{\alpha}_1 \bar{\beta}_1 \bar{\alpha}_1 \bar{\beta}_1 \bar{\alpha}_2 \bar{\beta}_2$
in $\bar{x}_0(\bar{\alpha}_1\bar{\beta}_1\bar{\alpha}_1\bar{\beta}_1\bar{\alpha}_2\bar{\beta}_2)^C\bar{z}_0$,
there exists
$\delta \in \mathbb{N}$ satisfying $|\delta| < C + \max(|\bar{x}_0|,|x_0'|)$ such that~:
\begin{equation}\label{equ:pre}
\begin{array}{l}
    \alpha_1 \beta_1 \alpha_1 \beta_1 \alpha_2 \beta_2[i]
      =
    \bar{\alpha}_1 \bar{\beta}_1 \bar{\alpha}_1 \bar{\beta}_1 \bar{\alpha}_2 \bar{\beta}_2[i + \delta]\\
    \text{ for all $\max(0,\delta) < i \leq
    \min(|\alpha_1 \beta_1 \alpha_1 \beta_1 \alpha_2 \beta_2|,
    |\alpha_1 \beta_1 \alpha_1 \beta_1 \alpha_2 \beta_2|-\delta)$.}
\end{array}
\end{equation}
Finally, as $\delay(\bar{u}_1\bar{u}_1',\bar{u}_2\bar{u}_2') = \delay(\bar{u}_1\bar{u}_1'\bar{v}_1,\bar{u}_2\bar{u}_2'\bar{v}_2)$,
\Cref{claim:TP} yields that
\begin{equation}\label{equ:STP}
    \bar{\alpha}_1[i] = \bar{\alpha}_2[i]
    \textup{ for every $1 \leq i \leq \min(|\bar{\alpha}_1|,|\bar{\alpha}_2|)$}.
\end{equation}
We now combine all equations.
For every $\max(0,\delta) < i \leq \min(|\alpha_1|,|\alpha_2|-\delta)$ we get~:
\[
\begin{array}{llll}
    \alpha_1[i]
    & = & \bar{\alpha}_1[i + \delta] & \text{ by~\eqref{equ:pre},}\\
    & = & \bar{\alpha}_2[i + \delta] & \text{ by~\eqref{equ:STP},}\\
    & = & (\bar{\alpha}_1
    \bar{\beta}_1 \bar{\alpha}_1
    \bar{\beta}_1 \bar{\alpha}_2)[|\bar{\alpha}_1
    \bar{\beta}_1\bar{\alpha}_1
    \bar{\beta}_1| + i + \delta] & \\
    & = & (\bar{\alpha}_1
    \bar{\beta}_1\bar{\alpha}_1
    \bar{\beta}_1 \bar{\alpha}_2)[|\alpha_1 \beta_1\alpha_1 \beta_1| + i + \delta] & \text{ by~\eqref{equ:loopSize},}\\
    & = & (\alpha_1 \beta_1\alpha_1 \beta_1 \alpha_2)[|\alpha_1 \beta_1\alpha_1 \beta_1| + i] & \text{ by~\eqref{equ:pre},}\\
    & = & \alpha_2[i].
\end{array}
\]

If $\delta < 0$ we still need to address some values of $i$. 
For every $1 \leq i \leq -\delta$ we get~:
\[
\begin{array}{llll}
    \alpha_1[i]
    & = & (\alpha_1 \beta_1 \alpha_1)[|\alpha_1 \beta_1| + i]\\
    & = & (\alpha_1 \beta_1 \alpha_1)[|\bar{\alpha}_1 \bar{\beta}_1| + i] & \text{ by~\eqref{equ:loopSize},}\\
    & = & (\bar{\alpha}_1 \bar{\beta}_1)[|\bar{\alpha}_1 \bar{\beta}_1| + i + \delta] & \text{ by~\eqref{equ:pre} since $i + \delta < 1$,}\\
    & = & (\bar{\alpha}_1 \bar{\beta}_1\bar{\alpha}_1 \bar{\beta}_1)[
    |\bar{\alpha}_1 \bar{\beta}_1\bar{\alpha}_1 \bar{\beta}_1| + i + \delta] \\
    & = & (\bar{\alpha}_1 \bar{\beta}_1\bar{\alpha}_1 \bar{\beta}_1)[|\alpha_1 \beta_1\alpha_1 \beta_1| + i + \delta] & \text{ by~\eqref{equ:loopSize},}\\
    & = & (\alpha_1 \beta_1 \alpha_1 \beta_1 \alpha_2)[|\alpha_1 \beta_1\alpha_1 \beta_1| + i] & \text{ by~\eqref{equ:pre} since $i \geq 1$,}\\
    & = & \alpha_2[i].
\end{array}
\]
These last two series of equations ensure that $\alpha_1$ and $\alpha_2$
have a large common prefix~:
\begin{equation}\label{equ:fin}
    \alpha_1[i] = \alpha_2[i] \text{ for every } 1 \leq i \leq \min(|\alpha_1|,|\alpha_2|-\delta).
\end{equation}
Then we can conclude by noticing that, thanks to the choice of 
\[
M > \frac{m \cdot |v_1| + C + \max(|\bar{x}_0|,|x_0'|)}{b \cdot |v_1|} > 
\frac{m}{b} + \frac{|\delta|}{b \cdot |v_1|},
\]
the words $\alpha_1 = u_0u_1v_1^{a + Mb}$ and $\alpha_2 = u_0u_2v_2^{a + Mb}$
can be written
as $u_0u_1v_1^{m}v_1^{\mu}$, respectively $u_0u_2v_2^{m}v_2^{\mu}$,
for some $\mu \in \mathbb{N}$ such that $|v_1^{\mu}| = |v_2^{\mu}| > |\delta|$.
Therefore, as desired, \Cref{equ:fin} implies that there is no mismatch between 
 $u_1v_1^m$ and $u_2v_2^m$.

\subparagraph*{For HTP} 
Suppose that $\calT$ satisfies the \HTP, and $\calS$ does not. Then, there exists a twinning pattern $(p_1,q_1,p_2,q_2,u,v,u_1,v_1,u_2,v_2)$
such that either
$(1)$ $|v_1|\neq |v_2|$ or
$(2)$ there is a position $i\in ]\text{min}(|u_1|,|u_2|), \text{max}(|u_1v_1|,|u_2v_2|]$ such that $(uv_1)[i]\neq (uv_2)[i]$.
In both cases,
we construct a family of input words $u$ for which the Hamming distance $d_h(\calS(uv_i^\lambda), \calT(u))$ gets arbitrarily large,
which contradicts the assumption. 

For case $(1)$, it is immediate, as iterating the loop creates outputs, on the same input, with arbitrarily large length differences.

For case $(2)$, assuming $|v_1|=|v_2|$, it is also immediate, because for all $j\geq 1$, $uv_1^j$ and $uv_2^j$ contain at least $j$ mismatching positions,
as for every $1 \leq k \leq j$,
\[
(uv_1^j)[i+(k-1)|v_1|] = (uv_1)[i] \neq (uv_2)[i]=(uv_2^j)[i+(k-1)|v_1|].\qedhere
\]
\end{proof}

We now show that checking each of the variants of the twinning property is decidable.
\begin{lemma}\label{lem:decideTPs}
    It is decidable whether a transducer satisfies \ATP{}, \STP{} and \HTP{}.
\end{lemma}

\begin{proof}
    Let $\calT$ be a transducer that defines a relation $R$.
    
    \subparagraph*{Deciding \ATP{}} Let $\calT^2$ be the cartesian product of $\calT$ by itself.
    Verifying whether $\calT$ satisfies \ATP{} reduces to checking
    that every loop in $\calT^2$ produces a pair of conjugate output words.
    For each state $(p,q) \in \calT^2$,
    we define a rational relation $R_{(p,q)}$
    consisting of pairs of output words produced by the loops in $\calT^2$ rooted at $(p,q)$.
    The transducer defining $R_{(p,q)}$ is obtained from $\calT^2$
    by disregarding the inputs and setting both the initial and final state to be $(p,q)$.
    It is known that we can decide whether every pair of words in a given rational relation is conjugate~\cite{decidingconjugacy}, which entails decidability of 
   \ATP.
   
    \subparagraph*{Deciding \STP{}} The twinning property is decidable in polynomial time~\cite{twinningproperty}. To decide \STP{}, we first decompose the transducer
    into SCCs (in polynomial time). Then, for each SCC and each state $p$,
    the SCC is seen as a transducer with initial state $p$ on which we check \TP{}.

    \subparagraph*{Deciding \HTP{}} We prove that it can be decided in \textsc{PTime} whether a transducer
    does not satisfy \HTP.
    The \HTP{} can be decomposed as a conjunction of two properties~:
    $\HTP_{lgth}$ on lengths and $\HTP_{mism}$ on mismatches.
    The (negation of) $\HTP_{lgth}$ can be directly expressed in the pattern logic of~\cite{DBLP:journals/ijfcs/FiliotMR20}, whose model-checking is in \textsc{PTime}. 
     Then, we reduce the problem of deciding the $\HTP_{mism}$ to the emptiness problem of a Parikh automaton of polynomial size, in the size of the transducer $\calT$. We recall that a Parikh automaton of dimension $d$ is an NFA extended with vectors in $\mathbb{Z}^d$ on its transitions. It accepts a word if there exists a run of the NFA that reaches an accepting state, and the sum of the vectors seen along the run belongs to some semi-linear set $W\subseteq \mathbb{Z}^d$, given for example as a Presburger formula $\phi(x_1,\dots,x_d)$. The emptiness problem is known to be decidable in \textsc{PTime} when both $d$ and $\phi$ are constant~\cite{DBLP:conf/lics/FigueiraL15,DBLP:journals/ijfcs/FiliotMR20}. Our reduction follows standard ideas, 
     and falls in this particular case. Let us give a bit more details.

     A twinning pattern $t$ can be encoded as a word $u_t$ over $\Delta^2\cup\{\#\}$, where $\Delta$ is the transition relation of $\calT$. In this construction, a run of $\calT$ is seen as a sequence of transitions. So, a twinning pattern consists of four runs, $r_1,r'_1$ and $r_2,r'_2$ where for all $i=1,2$, $r_i$ is the run such that $p_i\xrightarrow{u|u_i}q_i$ and $r'_i$ is the run such that $q_i\xrightarrow{v|v_i}q_i$.
     The four runs are encoded as a word $u_t = (r_1\otimes r_2)\# (r'_1\otimes r'_2)$ where $r_1\otimes r_2$ is the convolution of $r_1$ and $r_2$, i.e. the overlapping of $r_1$ and $r_2$ (and similarly for $r'_1\otimes r'_2$).
    It should be clear that the set of words $u_t$ for all twinning pattern $t$ such that $p_1$ and $p_2$ are initial states is a regular language, recognizable by an NFA of polysize in the size of $\calT$.

    In order to check the condition that there is a mismatch between $u_1v_1$ and $u_2v_2$ on a position common to $v_1$ and $v_2$, the NFA is extended with two counters $c_1$ and $c_2$, and the linear acceptance condition $c_1=c_2=0$ (making it a Parikh automaton). Those two counters are used to guess the mismatching position. Initially, using an $\epsilon$-loop, those two counters are incremented in parallel. At the end of this first phase, they therefore hold the same value, say $i\in\mathbb{N}$. Then, the Parikh automaton is built in such a way that it checks that $|u_1|<i\leq |u_1v_1|$ and $|u_2|<i\leq |u_2v_2|$, and $(u_1v_1)[i]\neq (u_2v_2)[i]$. To do so, while reading any transition of $r_j$ producing some word $\alpha_j$, $j=1,2$, $c_j$ is decremented by $|\alpha_j|$. The same is done when reading transitions of $r'_j$, but the automaton can non-deterministically guess that the counter $c_j$ is equal to $0$, and store (in its state) the corresponding letter in $\alpha_j$.
    From then on, it never decrements the counter $c_j$ again.
    When the whole input $(r_1\otimes r_2)\# (r'_1\otimes r'_2)$ has been read, the automaton has two letters stored in its state. It accepts the input, if those two letters are different, and if $c_1=c_2=0$.
     \end{proof}

\section{Approximate determinisation of rational functions}

In this section, we define \emph{approximately determinisable functions} and give decidable properties on transducers to characterise them.

\begin{definition}[Approximate determinisation]\label{def:approxdeter}
A rational function $f : A^\ast \rightarrow B^\ast$ is \emph{approximately determinisable} w.r.t.~metric $d$ if there exists a sequential function $g : A^\ast \rightarrow B^\ast$ such that $d(f,g) < \infty$. In this case, we also say that $g$ approximately determinises $f$ w.r.t.~$d$. 
\end{definition}

\begin{example}
The function $f_{\sf last}$ from the introduction is approximately determinisable w.r.t.~Levenshtein, while not w.r.t.~Hamming distance. The function $f_{\sf last}^\ast$ (depicted below) that maps $u_1\# \cdots u_n\#$, for $n \geq 1$, to $f_{\sf last}(u_1)\# \cdots f_{\sf last}(u_n)\#$ for some separator $\#$ is approximately determinisable w.r.t~neither Levenshtein nor Hamming distance.
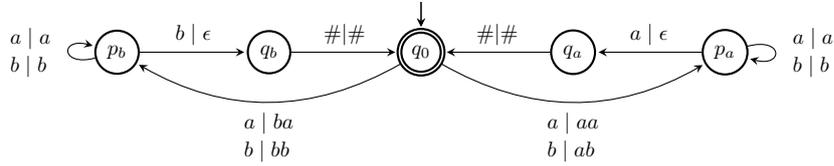
\begin{center}
\begin{tikzpicture}[draw=black,shorten >=1pt, auto, node distance=2cm, on grid, >=stealth,scale=0.8, every node/.style={scale=0.8}]
    \tikzset{state/.style={circle, draw=black, thick, minimum size=0.5cm}}
        \node[state, draw,accepting,initial,initial above, initial text=] (q0) {$q_0$};
        \node[state,draw] (q3) [right=of q0] {$q_a$};
        \node[state,draw] (q4) [left=of q0] {$q_b$};
        \node[state] (q1) [right=of q3] {$p_a$};
        \node[state] (q2) [left=of q4] {$p_b$};
    
        \path[->]
            (q0) edge [bend right,below] node {$\begin{array}{lll}a\mid aa\\ b\mid ab\end{array}$} (q1)
            (q0) edge [bend left,below] node {$\begin{array}{l}a\mid ba\\ b\mid bb\end{array}$} (q2)
            (q1) edge [loop right] node {$\begin{array}{l}a\mid a\\b\mid b\end{array}$} (q1)
            (q2) edge [loop left] node {$\begin{array}{l}a\mid a\\b\mid b\end{array}$} (q2)
            (q2) edge [above] node {$b\mid \epsilon$} (q4)
            (q1) edge [above] node {$a\mid\epsilon$} (q3)
            (q3) edge [above] node  {$\#|\#$} (q0)
            (q4) edge [above] node  {$\#|\#$} (q0);
    \end{tikzpicture}
    \end{center}

\end{example}
The approximate determinisation problem asks whether a rational function given as a functional transducer
is approximately determinisable.
We prove the following.
\begin{theorem}\label{theorem:main}
    The approximate determinisation problem for rational functions w.r.t.~Levenshtein family ($d_l,d_{lcs},d_{dl}$) and Hamming $(d_h)$ distance are decidable.
\end{theorem}
To prove the theorem, we show that \ATP{} and \STP{} characterise
rational functions that can be approximately determinised w.r.t~Levenshtein family (\Cref{theorem:approxdetlev}).
Similarly, we establish that \HTP{} and \STP{} characterise rational functions that can be approximately determinised
w.r.t~Hamming distance (\Cref{theorem:approxdetHamming}).
Theorem~\ref{theorem:main} then follows directly,
as these three properties are decidable (\Cref{lem:decideTPs}).
We now outline the proof strategy.
The full proof for Levenshtein family is presented in \cref{subsec:lev},
while the proof for Hamming distance, which follows a similar approach, is deferred to \cref{subsec:ham}.

\vspace{-0.3cm}
\subparagraph*{Levenshtein family}
We show with \Cref{proposition:approxtoATPSTP} that \ATP{} and \STP{}
are necessary conditions for approximate determinisation with respect to Levenshtein family,
a consequence of \Cref{lem:preserveTPs}.
The main challenge lies in proving that these properties are sufficient.
To prove it, we first show that \ATP{} alone
suffices for certain subclasses of functional transducers:
it enables the approximate determinisation of multi-sequential (\Cref{lemma:multiseqapprox})
and unambiguous series-sequential (\Cref{lemma:seriesseqapprox}) functions.
However, for rational functions in general, \ATP{} does not suffice. For example, the transducer above for $f^*_{\sf last}$ satisfies \ATP{} but is not approximately determinisable.
To extend this result to all rational functions,
we incorporate \STP{}.
Given a functional transducer $\trans$ satisfying \STP{},
we transform each strongly connected component of  $\trans$ into a sequential transducer,
effectively decomposing $\trans$ into a finite union of concatenations of sequential transducers.
We then leverage our results for series-sequential and multi-sequential functions
to approximate this structure with a sequential function~(\Cref{theorem:approxdetlev}).

\Cref{fig:enter-label} illustrates the main construction technique used in these proofs:
Starting with a transducer 
$\trans$ that we aim to approximate, we construct a sequential transducer 
$\calD_1$ as follows. We apply the powerset construction to 
$\trans$, introducing a distinguished state (marked by a $\bullet$ in the figure) in each subset.
The output is determined by the distinguished state’s production.
If the distinguished state reaches a point where it has no continuation,
we simply transition to another distinguished state.
We show that \ATP{}, combined with a carefully chosen priority scheme for selecting distinguished states,
ensures bounded Levenshtein distance.

\vspace{-0.3cm}
 \subparagraph*{Hamming Distance} 
The proof strategy is similar to the Levenshtein setting:
We first show with \Cref{proposition:approxtoHTPSTP}
that \HTP{} and \STP{} are necessary conditions
for approximate determinisation with respect to Hamming distance,
we show that \HTP{} alone suffices for the approximate determinisation
of multi-sequential (\Cref{lemma:multiseqapproxHamming}) and
series-sequential (\Cref{lemma:seriesseqapproxHamming}) functions,
and then we conclude by using \STP{} to transform functional transducers
into a finite unions of concatenations of sequential transducers (\Cref{theorem:approxdetHamming}).

While the core ideas remain similar to those used for the Levenshtein family,
the constructions required for the Hamming distance,
illustrated in \Cref{fig:enter-label},
are more intricate.  
Approximating the transducer $\trans$ with respect to the Levenshtein distance
(as shown by $\calD_1$ in the figure) allows us, at each step, to select a run,
produce its output, and disregard other possible runs.
However, for the Hamming distance, it is crucial to carefully track the length difference
between the produced output and the potential outputs of alternative runs.  
For instance, compare the outputs of $\trans$, $\calD_1$, and $\calD_2$ after reading $babcccc$~:
\[
\begin{array}{lll}
\calT(\texttt{babcccc}) & = & \texttt{bbabcabcabcabc},\\
\calD_1(\texttt{babcccc}) & = & \texttt{aabcabcabcabc},\\
\calD_2(\texttt{babcccc}) & = & \texttt{ababcabcabcabc}.
\end{array}
\]
We observe that after reading the input $bab$,
$\calD_1$ realizes that its distinguished state is incorrect and jumps to another state.
However, this shift causes a misalignment with $\trans$,
and reading additional $c$'s results in arbitrarily
many mismatches.\footnote{The Levenshtein distance remains bounded, as inserting a letter at the start resynchronizes both outputs.}  
In contrast, $\calD_2$ keeps in memory the delay relative to other runs.
Although it may still introduce mismatches along the way,
it ensures that when the distinguished run terminates,
it adjusts the output while transitioning to another run,
preventing long-term misalignment with $\trans$.

\begin{figure}
    \centering
    \begin{subfigure}{0.25\textwidth}
        \begin{tikzpicture}[draw=black,shorten >=1pt, auto, node distance=2cm, on grid, >=stealth,scale=0.7, every node/.style={scale=0.7}]
        \tikzset{state/.style={circle, draw=black, thick, minimum size=0.5cm}}
        \arraycolsep=1.4pt
            \node at (-0.75,0) {\Large$\calT$:}; 
            \node[state] (s) at (0.625,0) {$s$};
            \node[state] (p1) at (0,-1.5) {$p_1$};
            \node[state] (p2) at (0,-3) {$p_2$};
            \node[state] (q1) at (1.25,-1.5) {$q_1$};
            \node[state] (q2) at (1.25,-3) {$q_2$};
            \node[state,accepting] (f) at (0.625,-4.5) {$f$};
        
            \path[->]
                ($(s) + (0,0.7)$) edge (s)
                (s) edge [left] node[yshift=5pt] {$\begin{array}{lll}a &|& a\\b&|&\epsilon\end{array}$} (p1)
                (p1) edge [left] node {$\begin{array}{lll}a &|& a\\b&|&\epsilon\end{array}$} (p2)
                (p2) edge [left] node {$a \mid a$} (f)
                (s) edge [right] node[yshift=5pt] {$\begin{array}{lll}a &|& \epsilon\\b&|&b\end{array}$} (q1)
                (q1) edge [right] node {$\begin{array}{lll}a &|& \epsilon\\b&|&b\end{array}$} (q2)
                (q2) edge [right] node {$b \mid b$} (f)
                (f) edge [loop below] node[left,xshift=-5pt] {$c \mid abc$} (f);
        \end{tikzpicture}
    \end{subfigure}
    \begin{subfigure}{0.25\textwidth}
        \begin{tikzpicture}[draw=black,shorten >=1pt, auto, node distance=2cm, on grid, >=stealth,scale=0.8, every node/.style={scale=0.8}]
        \tikzset{state/.style={circle, draw=black, thick, minimum size=0.5cm}}
        \arraycolsep=1.4pt
            \node at (-1.25,0) {\Large$\calD_1$:}; 
            \node[draw,rectangle, rounded corners, minimum height = 18pt, thick] (p0)
            at (0,0) {\ $s$ \ };
            \node[draw,rectangle, rounded corners, minimum height = 18pt, thick] (p1)
            at (0,-1.5) {\ $p_1$ \ $q_1$ \ };
            \node[draw,rectangle, rounded corners, minimum height = 18pt, thick] (p2)
            at (0,-3) {\ $p_2$ \ $q_2$ \ };
            \node[draw,rectangle, rounded corners, minimum height = 18pt, thick] (p4)
            at (0,-4.5) {\ $f$ \ };
            \node[draw,rectangle, rounded corners, minimum height = 16pt, minimum width = 16, thick]
            at (0,-4.5) {};
            \node at ($(0,0) + (-0.17,0)$) {$\bullet$};
            \node at ($(0,-1.5) + (-0.56,0)$) {$\bullet$};
            \node at ($(0,-3) + (-0.56,0)$) {$\bullet$};
            \node at ($(0,-4.5) + (-0.17,0)$) {$\bullet$};
        
            \path[->]
                ($(p0) + (0,0.7)$) edge (p0)
                (p0) edge [left] node {$a \mid a\strut$} (p1)
                (p1) edge [left] node {$a \mid a\strut$} (p2)
                (p2) edge [left] node {$a \mid a\strut$} (p4)
                (p0) edge [right] node {$b\mid\epsilon\strut$} (p1)
                (p1) edge [right] node {$b\mid\epsilon\strut$} (p2)
                (p2) edge [right] node {$b\mid\epsilon\strut$} (p4)
                (p4) edge [loop below] node[left,xshift=-5pt] {$c \mid abc$} (p4);
        \end{tikzpicture}
    \end{subfigure}
    \begin{subfigure}{0.45\textwidth}
        \begin{tikzpicture}[draw=black,shorten >=1pt, auto, node distance=2cm, on grid, >=stealth,scale=0.8, every node/.style={scale=0.8}]
        \tikzset{state/.style={circle, draw=black, thick, minimum size=0.5cm}}
        \arraycolsep=1.4pt
            \node at (-1.25,0) {\Large$\calD_2$:}; 
            \node[draw,rectangle, rounded corners, minimum height = 18pt, thick] (p0)
            at (0,0) {\ $(s,\epsilon)$ \ };
            \node[draw,rectangle, rounded corners, minimum height = 18pt, thick] (p11)
            at (-1.35,-1.5) {\ $(p_1, a)$ \ $(q_1,\epsilon)$ \ };
            \node[draw,rectangle, rounded corners, minimum height = 18pt, thick] (p12)
            at (1.35,-1.5) {\ $(p_1, \epsilon)$ \ $(q_1,b)$ \ };
            \node[draw,rectangle, rounded corners, minimum height = 18pt, thick] (p21)
            at (-2.7,-3) {\ $(p_2, aa)$ \ $(q_2,\epsilon)$ \ };
            \node[draw,rectangle, rounded corners, minimum height = 18pt, thick] (p22)
            at (0,-3) {\ $(p_2, \epsilon)$ \ $(q_2,\epsilon)$ \ };
            \node[draw,rectangle, rounded corners, minimum height = 18pt, thick] (p23)
            at (2.7,-3) {\ $(p_2, \epsilon)$ \ $(q_2,bb)$ \ };
            \node[draw,rectangle, rounded corners, minimum height = 18pt, thick] (f)
            at (0,-4.5) {\ $(f,\epsilon)$ \ };
            \node[draw,rectangle, rounded corners, minimum height = 16pt, minimum width = 32, thick]
            at (0,-4.5) {};
            \node at ($(0,0) + (-0.47,0)$) {$\bullet$};
            \node at ($(-1.35,-1.5) + (-1.14,0)$) {$\bullet$};
            \node at ($(1.35,-1.5) + (-1.11,0)$) {$\bullet$};
            \node at ($(-2.7,-3) + (-1.21,0)$) {$\bullet$};
            \node at ($(0,-3) + (-1.10,0)$) {$\bullet$};
            \node at ($(2.7,-3) + (-1.19,0)$) {$\bullet$};
            \node at ($(0,-4.5) + (-0.47,0)$) {$\bullet$};
        
            \path[->]
                ($(p0) + (0,0.7)$) edge (p0)
                (p0) edge [left] node[yshift=2pt] {$a \mid \epsilon\strut$} (p11)
                (p0) edge [right] node[yshift=2pt] {$b \mid \epsilon\strut$} (p12)
                (p11) edge [left] node[yshift=2pt] {$a \mid \epsilon\strut$} (p21)
                (p11) edge [left] node[yshift=2pt,xshift=-2pt] {$b \mid a\strut$} (p22)
                (p12) edge [right] node[yshift=2pt,xshift=2pt] {$a \mid a\strut$} (p22)
                (p12) edge [right] node[yshift=2pt] {$b \mid \epsilon\strut$} (p23)
                (p21) edge [left, bend right] node[yshift = -5pt] {$a\mid aaa\strut$} (f)
                (p21) edge [right, bend right] node[yshift = 9pt, xshift = -9pt] {$b\mid b\strut$} (f)
                (p22) edge [left] node {$a\mid a\strut$} (f)
                (p22) edge [right] node {$b\mid b\strut$} (f)
                (p23) edge [left, bend left] node[yshift = 9pt, xshift = 9pt] {$a\mid a\strut$} (f)
                (p23) edge [right, bend left] node[yshift = -5pt] {$b\mid bbb\strut$} (f)
                (f) edge [loop below] node[left,xshift=-5pt] {$c \mid abc$} (f);
        \end{tikzpicture}
    \end{subfigure}
    \caption{An unambiguous non-deterministic transducer $\calT$,
    along with two sequential approximations $\calD_1$ and $\calD_2$
    with respect to the Levenshtein distance, respectively Hamming distance.}
    \label{fig:enter-label}
\end{figure}
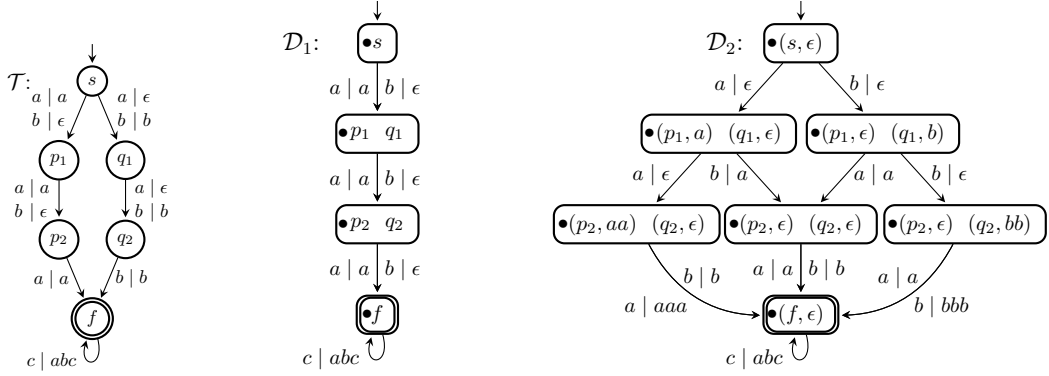

\subsection{Approximate determinisation for Levenshtein Family}\label{subsec:lev}
We give a decidable characterisation of approximately determinisable rational functions w.r.t.~Levenshtein family of distances --- Levenshtein ($\dlev$), Longest common subsequence ($d_{lcs}$) and Damerau-Levenshtein ($d_{dl}$). They are all equivalent up to boundedness (Lemma 2.1 and Remark 7 of \cite{editdistance}), i.e., for any two rational functions $f,g$, $d_{lcs}(f,g) < \infty \iff \dlev(f,g) < \infty \iff d_{dl}(f,g) < \infty.$ We show that a rational function is approximately determinisable w.r.t.~Levenshtein family if and only if the transducer that defines the function satisfies both \ATP{} and \STP{}. One direction is a consequence of \Cref{lem:preserveTPs} as follows.

\begin{proposition}\label{proposition:approxtoATPSTP}
   If a rational function given by a trim transducer $\calT$ is approximately determinisable w.r.t.~a metric $d \in \{d_l,d_{lcs},d_{dl}\}$ then $\calT$ satisfies both \ATP{} and \STP{}.
\end{proposition}
\begin{proof}
    Given that the rational function given by $\calT$ is approximately determinisable, i.e., there exists a sequential function given by a deterministic transducer $\calD$ such that $d(\calT,\calD) < \infty$. Since $\calD$ is a sequential transducer, it satisfies the twinning property, and consequently, $\calD$ also satisfies \ATP{} and \STP{} by \Cref{proposition:TPrelations}. Since $d(\calT,\calD) <\infty$, we can conclude that $\calT$ also satisfies \ATP{} as well as \STP{} by applying \Cref{lem:preserveTPs}.
\end{proof}

Towards proving the other direction, we prove the following lemma, which provides a bound on the distance between the output words produced by distinct runs of a transducer on the same prefix of an input word using {\Cref{prop:levboundconj}.

\begin{lemma}\label{lem:ATPtobound}
       Let $\calT$ be a trim transducer satisfying the \ATP{}. Then, there exists a constant $N_\calT \in \mathbb{N}$ such that for any two output words $v, v' \in B^*$ produced via two distinct runs of $\calT$ from an initial state on the same prefix of an input word, $d(v,v') \leq N_\calT$ for $d \in \{d_l,d_{lcs},d_{dl}\}$.
    \end{lemma}

    \begin{proof}
     Let $\calT^2$ be the cartesian product of $\calT$ by itself. By designating all states of $\calT^2$ as final, and disregarding the input word, we obtain a new transducer that defines the relation $R_o$, consisting of all pairs of output words produced by distinct runs of $T$ on the same prefix of an input word.  Since $\calT$ satisfies \ATP{}, every pair of output words produced by loops in $\calT^2$ on any input are conjugate. Consequently, since $R_o$ is obtained by ignoring the input word from $\calT^2$, it satisfies the condition in \Cref{prop:levboundconj}. Hence, the diameter of $R_o$ w.r.t.~Levenshtein family of distances is bounded. 
     
     Let $\dia{R_o}{d} \leq k$. By the definition of diameter, it follows that $d(u,v) \leq k$ for all $(u,v) \in R_o$. As a result, the distance between any two output words produced by distinct runs of $\calT$ on the same word is at most $k$. Setting $N_T = k$ completes the proof.
    \end{proof}

For subclasses of rational functions, namely multi-sequential and series-sequential functions, we show that \ATP{} is a sufficient condition for approximate determinisation. 
\begin{lemma} \label{lemma:multiseqapprox}
    A multi-sequential function given by a trim transducer $\calT$ is approximately determinisable w.r.t.~a metric $d \in \{d_l,d_{lcs},d_{dl}\}$ iff $\calT$ satisfies the \ATP{}.
\end{lemma}

\begin{proof}
    ($\rightarrow$) is direct by \Cref{proposition:approxtoATPSTP}, and we show ($\leftarrow$). Since the transducer $\calT$ is multi-sequential, it is equivalent to some finite union of sequential transducers ${\cal U} = \calD_1 \cup \cdots \cup \calD_k$ for some $k\in \mathbb{N}$ where $\calD_i = (Q_i,s_i,\delta_i,F_i,\lambda_i)$ is a sequential trim transducer for $i \in [k]$. By \Cref{lem:preserveTPs} and since $\calT$ satisfies \ATP{}, the transducer ${\cal U}$ also satisfies \ATP{}.    
    We construct a sequential transducer $\calD$ that approximately determinises $\calT$, which intuitively is simply the cartesian product of $\calD_1, \ldots, \calD_k$ which on each transition produces the output of the smallest index transducer $\calD_i$ for which that transition is defined.
    Let $\calD = (Q,s,\delta,F,\lambda)$  where
    \begin{enumerate}
        \item The set of states $Q = Q_1' \times Q_2' \times \cdots \times Q_k'$ is the cartesian product of the state set $Q_i'$ for $i \in [k]$ such that $Q_i'= Q_i \cup \{d\}$ where $d$ represents a dead state. The initial state is $s = (s_1,s_2,\cdots,s_k)$, and the set of final states $F$ is $\{(p_1, p_2, \ldots, p_k) \mid \exists i \  p_i \in F_i \}$.
 
        \item The function $\delta : Q \times A \rightarrow Q \times B^*$ is defined as: $\delta((p_1,p_2,\ldots,p_k),a)=((q_1,q_2,\ldots, q_k),x)$ where 
        for each $i \in [k]$, either $\delta_i(p_i,a) = (q_i,x_i)$ or, if $p_i = d$ or $\delta_i(p_i,a)$ is undefined, then $q_i = d$. The output $x$ is set to $x_j$, where $j \in [k]$ is the smallest index for which the transition $\delta_j(p_j,a)$ is defined.

        \item The output function $\lambda: F \rightarrow B^*$ is defined as $\lambda((p_1, p_2, \ldots, p_k))=\lambda_i(p_i)$ where $i \in [k]$ is the smallest index such that $p_i \in F_i$.
    \end{enumerate}
 
    We show that $d(\calD,\calT) < \infty$ w.r.t.~Levenshtein family. From the construction, it is clear that $\calD$ and $\calT$ have the same domain. Consider an input word $u\in\dom(\calT)$. 
    Let $i\in[k]$ be the smallest index such that $u\in\dom(\calD_i)$, with output word, say $v$. The transducer $\calD$ on input $u$ produces output, say $v'$, by concatenating output on each transition over input word produced by the smallest index transducer.
    
    Thus, $v'$ can be decomposed into $v_{i_1}v_{i_2} \cdots v_{i_n}$ where $i_1 < i_2 < \cdots < i_n = i$ and for each $i_j \in [i]$, $v_{i_j}$ is the output produced by $\calD_{i_j}$ along the partial run of $u$. For each $v_{i_j}$, let $v_{i_j}'$ denote the prefix of the output produced by $\calD_{i_j}$ upto $v_{i_j}$. The output produced by $\calT$ on $u$ via $\calD_i$ is $v = v_i'v_i = v'_{i_n}v_{i_n}$.

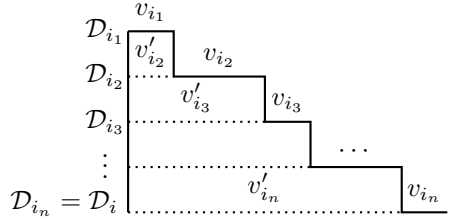
    \begin{center}
    \begin{tikzpicture}[scale=0.6]
        \draw[thick] (0,4) -- (1,4) node[midway, above] {$v_{i_1}$} 
                     -- (1,3)  
                     -- (3,3) node[midway, above] {$v_{i_2}$} 
                     -- (3,2)  
                     -- (4,2) node[midway, above] {$v_{i_3}$} 
                     -- (4,1)  
                     -- (6,1) node[midway, above] {$\cdots$}
                     -- (6,0)  
                     -- (7,0) node[midway, above] {$v_{i_n}$};
        \draw[thick] (0,0) -- (0,4);
        \draw[thick,dotted] (0,3) -- (1,3) node[midway, above] {$v'_{i_2}$};
        \draw[thick,dotted] (0,2) -- (3,2) node[midway, above] {$v'_{i_3}$};
        \draw[thick,dotted] (0,1) -- (4,1);
        \draw[thick,dotted] (0,0) -- (6,0) node[midway, above] {$v'_{i_n}$};

        \node at (-0.5,4) {$\calD_{i_1}$};
        \node at (-0.5,3) {$\calD_{i_2}$};
        \node at (-0.5,2) {$\calD_{i_3}$};
        \node at (-0.5,1.2) {$\vdots$};
        \node at (-1.4,0.2) {$\calD_{i_n} = \calD_i$};
    \end{tikzpicture}
\end{center}
    
 Observe that $v_{i_j}'v_{i_j}$ and $v_{{i_j}+1}'$ (for $i_1 < i_j < i_n$) is the output produced by $\calD_{i_j}$ and $\calD_{{i_j}+1}$ on the same prefix of input $u$. By \Cref{lem:ATPtobound}, we obtain $d(v_{i_j}'v_{i_j},v_{{i_j}+1}') \leq N_T$. Similarly,  since $v_{i_1},v_{i_2}'$ are the outputs of $\calD_{i_1}$ and $\calD_{i_2}$ on the same input prefix, we get $d(v_{i_1},v_{i_2}') \leq N_T$.
    
    \begin{align*}
     d(v_{i_1}v_{i_2}v_{i_3} \ldots v_{i_n},v) &= d(v_{i_1}v_{i_2}v_{i_3} \cdots v_{i_n},v_{i_n}'v_{i_n}) \text{ (Since $v= v_{i_n}'v_{i_n}$)}\\
     &\leq d(v_{i_1}v_{i_2}v_{i_3} \cdots v_{i_n},v_{i_2}'v_{i_2}v_{i_3} \cdots v_{i_n}) + d(v_{i_2}'v_{i_2}v_{i_3} \cdots v_{i_n},v_{i_3}'v_{i_3} \cdots v_{i_n})\\
     & \hspace{0.6cm} + \cdots + d(v_{{i_n}-1}'v_{{i_n}-1}v_{i_n},v_{i_n}'v_{i_n}) \text{ (Applying triangle inequality of $d$)}\\
     &\leq n \cdot N_T  \text{ (using \Cref{lem:ATPtobound})}   
    \end{align*}
    In fact, on any input word, the distances between the outputs produced by $\calD$ and $\calT$ is less than or equal to $k \cdot N_T$, 
    as there can be at most $k$ switches between runs.
    Hence, we get that $d(\calD,\calT)$ is bounded.
\end{proof}

The characterisation of Lemma~\ref{lemma:multiseqapprox} also holds for series-sequential functions.
\begin{lemma} \label{lemma:seriesseqapprox}
    Let $\calT = \calD_1 \cdots \calD_k$  be an unambiguous transducer where each $\calD_i$ is a sequential trim transducer for $i \in [k], k>1$. The series-sequential function defined by $\calT$ is approximately determinisable w.r.t.~a metric $d \in \{d_l,d_{lcs},d_{dl}\}$ if and only if $\calT$ satisfies \ATP{}.
\end{lemma}

\begin{proof}
    ($\rightarrow$) follows from \Cref{proposition:approxtoATPSTP}, and we prove
    ($\leftarrow$).
    Similarly to the proof in \Cref{lemma:multiseqapprox}, we construct a sequential transducer $\calD$ using a subset construction of $\calT$ such that $\calD$ approximately determinises $\calT$. Although $\calT$ is functional, and each $\calD_i$, $i \in [k]$, is sequential, the nondeterminism arises between the transitions between one $\calD_i$ to the other. We assume an ordering for the states of $\calT$ for each $D_i$. Intuitively, $\calD$ stores a subset of states of $\calT$ to capture all possible \emph{active runs} on any input word. It produces the output of the active run of the smallest indexed sequential transducer, called the \emph{producing run}. Upon termination of the producing run, it switches to the next active run of the smallest indexed sequential transducer. Since $\calT$ is unambiguous, there is at most one accepting run on any input. Formally, $\calD = (Q,s_0,\delta,F,\lambda)$  where
    \begin{enumerate}
        \item $Q$ is the power set of the states in $\calT$. Each set $S \in Q$ has exactly one state marked with a $\bullet$ to denote the currently producing run.
        \item $s_0$ is the initial state of $\calT$ and is marked with a $\bullet$.
        \item The transition function $\delta: Q \times A \rightarrow Q \times B^*  $ is defined as follows: $\delta(P,a)=(S,x)$ where $S$ consists of all states $q$ in $\calT$ reachable from a state $p \in P$ via transition $(p,a,q,x_{pq})$ in $\calT$. The output word $x$ is set as follows. Let $p$ be the $\bullet$ marked state in $P$ that belongs to the sequential transducer $\calD_i$ for some $i \in [k]$.
        \begin{enumerate}
            \item If a transition $(p,a,q,x_{pq})$ in $\calT$ exists within the same sequential transducer $\calD_i$, then $x=x_{pq}$ and $q$ is $\bullet$ marked in $S$. 
            \item If no transition exists from $p$ on $a$ within $\calD_i$, but a transition $(p,a,q,x_{pq})$ in $\calT$ exists to a different sequential transducer, i.e., $q$ belongs to $\calD_{i+1}$, then $x=x_{pq}$ and $q$ is $\bullet$ marked in $S$. Such a transition is called a \emph{switch} (to $\calD_{i+1}$).
            \item Otherwise, choose the smallest numbered state $p' \in P$ that belongs to the smallest indexed sequential transducer 
            $\calD_j$ where $j \in [k]$ with a defined transition $(p',a,q',x_{p'q'})$  in $\calT$, and set $x = x_{p'q'}$ and $\bullet$ mark $q'$  in $S$. Such a transition is also called a switch (to $\calD_j$). 
        \end{enumerate}
        \item $F$ is the set of all states $P \in Q$ such that $P$ contains a final state of $\calT$.
        \item The output function $\lambda: F \rightarrow B^*$ is defined, for $P \in F$, as $\lambda(P)=\lambda_\calT(p)$ for some arbitrary final state $p\in P$ of $\calT$, where $\lambda_\calT$ is the output function of $\calT$.
    \end{enumerate}
     The number of states of $\calD$ is exponential in the number of states of $\calT$. We now argue that the number of switches in the run of $\calD$ between the active runs of $\calT$ on any input word is less than the number of states in $\calT$, and hence finite. 
     Let $n_i$ be the number of states in $\calD_i$ for $i \in [k]$. Consider the run of $\calD$ on an arbitrary input word. Initially, only the initial state of $\calD_1$ is active in $\calD$.  
     As the run proceeds, if $P$ is the set of states of $\calT$ reached so far, then by construction of $\calD$, the only $\bullet$ marked state in $P$ is the last state of the active run of $\calT$ that belongs to the smallest indexed $\calD_i$. If this run eventually dies, then two cases can happen: $(1)$ $\calD$ switches to another active run in $\calD_i$, or, $(2)$ if none exists, $\calD$ switches to some active run in $\calD_j$, where $j>i$ is minimal. 
     If case $(2)$ happens, then no more states of $\calD_i$ are active, so $\calD$ will never switch again to $\calD_i$ in the future. The number of times case $(1)$ can happen is bounded by $n_i$. Indeed, at most $n_i$ states of $\calD_i$ can be active at any moment, and since $\calD_i$ is sequential, the number of active states in $\calD_i$ can only decrease.

     Now, if $\calD$ eventually switches to $\calD_j$, at most $n_j$ states in $\calD_j$ are active, so at most $n_j$ switches of type $(1)$ can happen in $\calD_j$, and so on until $\calD$ eventually terminates in some sequential transducer $\calD_l$ for $l\geq j$. 

Therefore, in the worst case, $\calD$ switches $n_i$ times in $\calD_i$ before switching to $\calD_{i+1}$ for all $i\in\{1,\dots,k-1\}$. So, the overall number of switches in the run of $\calD$ is at most $N = \sum_{i=1}^{k} n_i$, which is exactly the number of states in $\calT$.

    Now we show that $d(\calD,\calT) < \infty$ w.r.t.~Levenshtein family. From the construction, it is clear that $\calD$ and $\calT$ have the same domain.
   
    Consider an input word $u$ accepted by $\calT$. Since $\calT$ is unambiguous,
    there is exactly one run in $\calT$ that accepts $u$.
    Let $v = \calT(u)$ and $v' = \calD(u)$.
    From the construction of $\calD$, $v'$ can be decomposed into $v_{s_1}v_{s_2} \cdots v_{s_n}$ accommodating $n \in [N]$ switches between the active runs of $\calT$ on the prefixes of $u$, where each $v_{s_i}$, $i \in [n]$ is the output produced by an active run $r_i$ on the prefix of $u$. For each $v_{s_i}$, let $v_{s_i}'$ denote the prefix of the output produced by the run $r_i$ up to $v_{s_i}$. 
    
    Observe that $v_{s_i}'v_{s_i}$ and $v'_{s_{i+1}}$ (for $1 < i < n$) is the output produced by $\calD_1 \cdots \calD_{j_i}$ and $\calD_1 \cdots \calD_{j_{i+1}}$ on the same prefix of $u$ where $j_i,j_{i+1} \in [k]$. 
    By using \Cref{lem:ATPtobound}, we obtain $d(v_{s_i}'v_{s_i},v_{s_{i+1}}') \leq N_T$. Similarly,  since $v_{s_1},v_{s_2}'$ are the outputs of $\calD_1 \cdots \calD_{j_1}$ and $\calD_1 \cdots \calD_{j_2}$ on the same input prefix, it follows that $d(v_{s_1},v_{s_2}') \leq N_T$. As a consequence,
     \begin{align*}
     d(v_{s_1}v_{s_2}& \ldots v_{s_n},v) \\
     &= d(v_{s_1}v_{s_2} \cdots v_{s_n},v_{s_n}'v_{s_n}) \text{ ($v=v_{s_n}'v_{s_n}$ since $r_{s_n}$ is the only accepting run of $u$)}\\
     &\leq d(v_{s_1}v_{s_2} \cdots v_{s_n},v_{s_2}'v_{s_2} \cdots v_{s_n}) + d(v_{s_2}'v_{s_2} \cdots v_{s_n},v_{s_3}'v_{s_3} \cdots v_{s_n})\\
     & \hspace{0.6cm} + \cdots + d(v_{s_{n-1}}'v_{s_{n-1}}v_{s_n},v_{s_n}'v_{s_n}) \text{ (Applying triangle inequality of $d$)}\\
     &\leq n \cdot N_\calT \text{ (using \Cref{lem:ATPtobound})}
    \end{align*}
Therefore, $d(\calD(u),\calT(u))\leq n\cdot N_T\leq N\cdot N_\calT$. This holds for any $u\in\dom(\calT)$, and $N$ being the number of states in $\calT$, we get that $d(\calD,\calT)$ is bounded.
\end{proof}
We extend the characterisation to rational functions, where \STP{} is also required to decompose the function into a finite union of series-sequential functions, which can then be transformed into a multi-sequential function using the properties of \ATP{}. 

\begin{lemma}\label{theorem:approxdetlev}

    A rational function given by a trim transducer $\calT$ is approximately determinisable w.r.t.~a metric $d \in \{d_l,d_{lcs},d_{dl}\}$ iff $\calT$ satisfies both \ATP{} and \STP{}.
\end{lemma}

\begin{proof}
    ($\rightarrow$) follows from \Cref{proposition:approxtoATPSTP}. We prove ($\leftarrow$). Assume that the rational function is given by an unambiguous transducer $\calT$ with set of states $Q$. Disregarding the labels on transitions, decompose $\calT$ into maximal SCCs $S_1,\dots,S_k\subseteq Q$.
    Consider the set of paths $\Pi$ of the form $\pi = S_{i_1}t_{i_1}S_{i_2}\dots t_{i_{n-1}}S_{i_n}$ such that $S_{i_1}$ is an SCC which contains an initial state, $S_{i_n}$ is an SCC which contains a final state, and for all $1\leq k<n$, $t_{i_k}$ is a transition of $\calT$ from a state of $S_{i_k}$ to some state of $S_{i_{k+1}}$. Let $\calT_\pi$ denote the trim subtransducer of $\calT$ that removes all the transitions in $\calT$ except the transitions $t_{i_k}$ $(1\leq k<n)$ and the transitions occurring in the SCCs $S_{i_k}$ ($1\leq k\leq n$).

      Note that since the SCCs are maximal, the set $\Pi$ is finite. 
    Now, it is straightforward to see that $\calT \equiv {\cal U}= \bigcup_{\pi \in \Pi} \calT_\pi$. From \Cref{lem:preserveTPs}, we deduce ${\cal U}$ satisfies both \ATP{} as well as \STP{}. Moreover, given $\calT$'s unambiguity, each input accepted by $\calT$ is accepted by exactly one $\calT_\pi$ and, each SCC within a $\calT_\pi$ has a single entry and exit point.

    Since ${\cal U}$ satisfies \STP{}, each SCC in $T_\pi$ satisfies \TP{}, with initial state being the unique entry point of the SCC. Hence we can determinise each SCC in $\calT_\pi$ and can obtain a series-sequential transducer that is equivalent to $\calT_\pi$ and indeed satisfies \ATP{} by \Cref{lem:preserveTPs}. From \Cref{lemma:seriesseqapprox}, there exists a sequential transducer $\calD_\pi$ for each $\calT_\pi$, such that $d(\calD_\pi,\calT_\pi)$ is finite. 
    
    Let $d(\calT_\pi,\calD_\pi) \leq k_\pi$ for some $k_\pi \in \mathbb{N}$. Consequently, the distance between $\calT$ and the new transducer ${\cal U'} = \bigcup_{\pi\in \Pi} \calD_\pi$ is bounded where $d(\calT,{\cal U'}) = d({\cal U},{\cal U'}) \leq \max \{ k_\pi \mid \pi\in \Pi\}$. Further, since $\calT$ satisfies \ATP{}, we deduce ${\cal U'}$ also satisfies \ATP{} by \Cref{lem:preserveTPs}. Being multi-sequential and satisfying \ATP{}, ${\cal U'}$ can be further approximately determinised using the construction outlined in \Cref{lemma:multiseqapprox}. Let $\calD$ be the sequential transducer such that $d({\cal U'},\calD) < \infty$. Since $d$ is a metric,
    $d(\calT,\calD) \leq d(\calT,{\cal U'}) + d({\cal U'},\calD)$. Since both $d(\calT,{\cal U'})$ and $d({\cal U'},\calD)$ are finite, $d(\calT,\calD)$ is also finite. \qedhere

\end{proof}

\subsection{Approximate determinisation for Hamming distance}\label{subsec:ham}
In this subsection, we prove that a rational function is approximately determinisable w.r.t.~Hamming distance if and only if the transducer that defines the function satisfies both \HTP{} and \STP{}. The proof strategy is similar to the setting of Levenshtein family.
    \begin{lemma}\label{claim:HP}
    Let $T$ be a trim transducer satisfying the \HTP{}. Then, there exists a constant $N_T \in \mathbb{N}$ such that for any two output words $w, w' \in B^*$ produced by $T$ while processing the same input word via two distinct runs from an initial state, the following properties hold:
    \begin{enumerate}
        \item\label{item:sameLength} The length difference satisfies $\big||w| - |w'|\big| < N_T$.
        \item\label{item:mismatch} The number of positions where $w$ and $w'$ differ is at most $N_T$.
    \end{enumerate}
    \end{lemma}

    \begin{proof}
    We establish the lemma by decomposing the runs producing $w$ and $w'$ into short simple paths combined
    with synchronised looping segments.
    This decomposition naturally gives rise to multiple occurrences of the Twinning pattern.
    Since the transducer $T$ satisfies the \HTP{},
    we can leverage this property to show that all synchronised loops
    generate output words of the same length.
    Consequently, any difference in output length can only arise from the short simple paths.
    Furthermore, the \HTP{} also constrains the locations of potential mismatches between the two output words.
    By quantifying these effects, we derive the desired bounds on both the length difference
    and the number of mismatched positions.
    
    Formally, let $u \in A^*$ be such that $T$ has two initial runs $\rho$ and $\rho'$ on $u$,
    producing the outputs $w$ and $w'$, respectively. 
    Let $n \in \mathbb{N}$ denote the number of states of $T$,
    and let $M \in \mathbb{N}$ be the maximal output length produced by the transition function of $T$.
    We set $N_T = Mn^4$.
    
    Since there are at most $n^2$ distinct pairs of states in $T$,
    if the length of $u$ exceeds $n^2$, then the runs $\rho$ and $\rho'$ must revisit at least one pair of states while processing $u$, thus looping synchronously.
    More precisely, there exists a decomposition of $u$ into nonempty words
    \[
        u = u_0 v_1 u_1 \dots v_k u_k
    \]
    such that $|u_0 u_1 \dots u_k| < n^2$ and both $\rho$ and $\rho'$ loop
    while processing each $v_i$ for $1 \leq i \leq k$.
    This decomposition induces $k$ twinning patterns in $T$, as illustrated below:
    
    
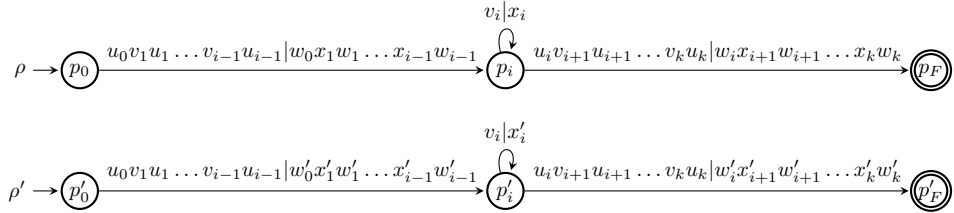
\begin{center}
    \begin{tikzpicture}[draw=black,shorten >=1pt, auto, node distance=2cm, on grid, >=stealth,scale=0.8, every node/.style={scale=0.8}]
        \tikzset{state/.style={circle, draw=black, thick, minimum size=0.5cm}}

        \node[state,inner sep = 6pt,initial,initial text=$\rho$] at (-7,1) (q_0){}; 
        \node[state,inner sep = 6pt,initial,initial text=$\rho'$] at (-7,-1) (q_0'){}; 
        \node[state,inner sep = 6pt] at (0,1) (q_i) {};
        \node[state,inner sep = 6pt] at (0,-1) (q_i') {};
        \node[state,inner sep = 6pt,accepting] at (7,1) (q_f) {};
        \node[state,inner sep = 6pt,accepting] at (7,-1) (q_f') {};
        
        \node at (q_0)   {$p_0$}; 
        \node at (q_0')   {$p_0'$}; 
        \node at (q_i)   {$p_i$}; 
        \node at (q_i')   {$p_i'$}; 
        \node at (q_f)   {$p_{F}$}; 
        \node at (q_f')   {$p_{F}'$}; 
        
        \path[->] 
        (q_0) edge [above] node {$u_0 v_1 u_1 \dots v_{i-1} u_{i-1} | w_0 x_1 w_1 \dots x_{i-1} w_{i-1}$} (q_i)
        (q_0') edge [above] node {$u_0 v_1 u_1 \dots v_{i-1} u_{i-1} | w_0' x_1' w_1' \dots x_{i-1}' w_{i-1}'$} (q_i')
        (q_i) edge [above] node  {$u_i v_{i+1} u_{i+1} \dots v_k u_k | w_i x_{i+1} w_{i+1} \dots x_k w_k$} (q_f)
                edge [loop above] node {$v_i | x_i$} ()
        (q_i') edge [above] node  {$u_i v_{i+1} u_{i+1} \dots v_k u_k | w_i' x_{i+1}' w_{i+1}' \dots x_k' w_k'$} (q_f')
                 edge [loop above] node {$v_i | x_i'$} ();
    \end{tikzpicture}
    \end{center}
    
    Since $T$ satisfies the \HTP{}, we have that $|x_i| = |x_i'|$ for every $1 \leq i \leq k$.
    This allows us to bound the difference in output length between $\rho$ and $\rho'$
    after reading each $u_i$:
    
    \begin{equation} \label{equ:sizeBound}
        \begin{aligned}
            \big| |w_0 x_1 w_1 \dots x_{i} w_{i}| - |w_0' x_1' w_1' \dots x_{i}' w_{i}'| \big|
            &= \big| |w_0 w_1 \dots w_{i}| - |w_0' w_1' \dots w_{i}'| \big| \\
            &\leq \max \big( |w_0 w_1 \dots w_{i}|, |w_0' w_1' \dots w_{i}'| \big) \\
            &\leq M \cdot |u_0 u_1 \dots u_{i}| \\
            &< Mn^2.
        \end{aligned}
    \end{equation}
    Applying Equation~\eqref{equ:sizeBound} to the case $i = k$ immediately yields Item~\ref{item:sameLength}:
    \[
        \big| |w| - |w'| \big| = \big| |w_0 x_1 w_1 \dots x_{k} w_{k}| - |w_0' x_1' w_1' \dots x_{k}' w_{k}'| \big| < Mn^2.
    \]
    To establish Item~\ref{item:mismatch},
    note that Equation~\eqref{equ:sizeBound} implies that for each $1 \leq i \leq k$,
    the subwords $x_i$ and $x_i'$ in $w$ and $w'$ are shifted by at most $Mn^2$.
    In other words, at most $Mn^2$ positions of $x_i$ don't overlap $x_i'$.
    As $T$ satisfies the \HTP{}, every position of $x_i$ that overlaps with $x_i'$
    is a match between $w$ and $w'$.
    Therefore, the positions where $w$ differs from $w'$ can only occur~:
    \begin{itemize}
        \item In one of the $w_i$,
        contributing at most $|w_1 w_2 \dots w_k| \leq M \cdot |u_1 u_2 \dots u_k| < Mn^2$ mismatches.
        \item In the part of some $x_i$ that does not overlap $x_i'$, contributing at most $k \cdot Mn^2$ mismatches.
    \end{itemize}
    As the $u_i$ are nonempty we can bound the number of mismatches between $w$ and $w'$ by
    \[
        Mn^2 + k \cdot Mn^2 \leq Mn^2 + (|u_0u_1 \ldots u_k|-1) Mn^2 < Mn^4.\qedhere
    \]
\end{proof}

\begin{proposition}\label{proposition:approxtoHTPSTP}
   If a rational relation given by a transducer $\calT$
   is approximately determinisable w.r.t.~Hamming distance
   then $\calT$ satisfies both \HTP{} as well as \STP{}.
\end{proposition}
\begin{proof}
    Let us suppose that the rational relation given by $\calT$
    is approximately determinisable, i.e.,
    there exists a sequential function given by a deterministic transducer $\calD$ such that $\dHam{\calT}{D} < \infty$.
    Since $\calD$ is a sequential transducer, it satisfies the twinning property, and consequently,
    $\calD$ also satisfies \HTP{} and \STP{} by \Cref{proposition:TPrelations}.
    Since $\dHam{\calT}{\calD} <\infty$, we can conclude that $\calT$ also satisfies \HTP{} as well as \STP{}
    by applying \Cref{lem:preserveTPs}.
\end{proof}

\begin{lemma} \label{lemma:multiseqapproxHamming}
    A multi-sequential function given by a transducer $\calT$
    is approximately determinisable w.r.t. the~Hamming distance
    if and only if $\calT$ satisfies the \HTP{}.
\end{lemma}

\begin{proof}
    ($\rightarrow$) By \Cref{proposition:approxtoHTPSTP},
    every transducer defining a relation that is
    approximately determinisable 
    with respect to the Hamming distance satisfies the \HTP{}.
    
    ($\leftarrow$) Let $\calT$ be a multi-sequential transducer satisfying the \HTP{}.
    Then $\calT$
    it is equivalent to a finite union of sequential transducers
    $\calU = \calD_1 \cup \calD_2 \cup \cdots \cup \calD_k$,
    where $\calD_i = (Q_i,s_i,\delta_i,F_i,\lambda_i)$
    is a sequential trim transducer for all $1 \leq i \leq k$.
    As $\calT$ satisfies the \HTP{},
    so does $\calU$ by \Cref{lem:preserveTPs}.
    We build a sequential transducer $\calD$ approximating $\calU$
    by modifying the \emph{subset construction with delays}
    used to determinise transducers satisfying the Twinning property~\cite{twinningproperty,twinningproperty0}.
    
    While processing an input word $u$,
    the transducer $\calD$ simulates the run of each $\calD_i$ on $u$,
    and tracks the outputs they produce.
    To that end, the states of $\calD$ are tuples of pairs
    \[
        ((p_1,v_1),(p_2,v_2), \ldots, (p_k,v_k)) \in
        ((Q_1 \cup \{d\}) \times B^*) \times ((Q_2\cup \{d\}) \times B^*) \times \ldots \times ((Q_k\cup \{d\}) \times B^*).
    \]
    Each $p_i$ is the state reached by $\calD_i$
    on the current input, and the special ``dead'' state $d$ is used to denote the fact that $\calD_i$
    does \emph{not} have a run on the current input.
    Each $v_i$ is a suffix of the output $v_i'$ produced by $\calD_i$
    on the current input.
    Observe that the complete output word $v_i'$ 
    may grow arbitrarily large as $u$ increases in length,
    and $\calD$ cannot store it entirely while maintaining a finite state space.
    $\calD$ decides which suffix to keep in memory
    through the following procedure~:
    To process an input letter $a$,
    $\calD$ updates each pair $(p_i,v_i)$ such that $\delta_i(p_i,a) = (q_i,w_i)$ into 
    $(q_i,v_iw_i)$.
    Then, $\calD$ identifies the \emph{minimal length} $C \in \mathbb{N}$ among all $|v_iw_i|$,
    truncates the prefix of length $C$ of all these words,
    and outputs one of the prefixes.
    While this approach may introduce errors, as the partial outputs produced along the run
    may not correspond to the run that is eventually accepting,
    our analysis shows that the satisfaction of the \HTP{}
    ensures bounded Hamming distance between the generated output
    and the correct one.
    
    Finally, at the end of the input word,
    $\calD$ selects a pair $(p,w)$
    from its current state such that $p$ is a final state,
    and outputs the concatenation of $w$ with the final output of $\trans$ at $p$.

    \subparagraph*{Formal construction}
    We set $\calD = (Q,s,\delta,F,\lambda)$, where
    \begin{enumerate}
        \item The set of states is $Q = Q_1' \times Q_2' \times \cdots \times Q_k'$,
        with
        $Q_i' = (Q_i \cup \{d\}) \times B^*$.
        While this definition allows arbitrarily large words,
        we will show that only a finite subset of $Q$ is reachable from the initial state.
        \item
        The initial state is $s = ((s_1,\epsilon),(s_2,\epsilon),\cdots,(s_k,\epsilon))$.
        \item The transition function maps each pair $(((p_1,v_1), (p_2,v_2), \ldots, (p_k,v_k)),a) \in Q \times A$ to
        \[
            \delta(((p_1,v_1), (p_2,v_2), \ldots, (p_k,v_k)),a)
            = (((q_1,x_1'), (q_2,x_2'), \ldots, (q_k,x_k')),w),
        \]
        with the pairs $(q_i,x_i')$ and the word $w$ defined as follows.
        First, for all $1 \leq i \leq k$ we set
        \[
        (q_i,x_i) =
        \begin{cases}
        (q_i,v_iw_i) & \text{if } p_i \in Q_i \text{ and } \delta_i(p_i,a) = (q_i,w_i); \\
        (d,\epsilon) & \text{if } p_i = d \text{, or if } p_i \in Q_i \text{ and } \delta_i(p_i,a) \text{ is undefined.}
        \end{cases}
        \]
        Then, we let $C \in \mathbb{N}$ be the minimal element of the set 
        $\{|x_{i}| \mid 1 \leq i \leq k \textup{ and } q_i \neq d\}$,
        or $C = 0$ if this set is empty.
        We set $w = x_j$, where $j$ is the smallest index satisfying $|x_j| = C$,
        and for every $1 \leq i \leq k$
        we set $x_i'$ as the word obtained by deleting the $C$ first letters of $x_i$.
        \item The set of final states is 
        $F = \{((p_1,u_1), (p_2,u_2), \ldots, (p_k,u_k)) \mid p_i \in F_i \textup{ for some } 1 \leq i \leq k\}$.
        \item The output function maps each state $p = ((p_1,u_1), (p_2,u_2), \ldots, (p_k,u_k)) \in Q$
        to
        \[
        \lambda(p) =
        \begin{cases}
        u_i \lambda_i(p_i) &  \text{if } p \in F \text{, where $i$ is the smallest index such that } p_i \in F_i; \\
        \epsilon & \text{otherwise}.
        \end{cases}
        \]
    \end{enumerate}

    \subparagraph*{Key properties}
    In order to use our construction to prove the Lemma,
    we explicitly state three key properties.
    Let us fix an input word $u \in A^*$,
    and let $p = ((p_1,u_1), (p_2,u_2), \ldots, (p_k,u_k)) \in Q$ and $v \in B^*$
    denote the state reached and the word produced by $D$ while reading $u$.
    First, observe that if we disregard the output words in the states of $D$
    and remove any occurrences of the ``dead'' state $d$,
    the resulting structure is exactly a subset construction keeping track of which
    states of $T$ can be reached on the current input.
    This leads to the following claim~:
    
    \begin{claim}\label{claim:states}
        For every $1 \leq i \leq k$ we have $q_i \neq d$
        if and only if $\calD_i$ has an initial run on $u$ that ends in $q_i$.
    \end{claim}
    Next, our construction keeps the output suffixes $u_i$
    synchronised with the outputs of $\trans$:

    \begin{claim}\label{claim:words}
        For every $i,i'$ such that $p_i \neq d$ and $p_{i'} \neq d$,
        the length difference $|u_i| - |u_i'|$ is equal
        to the length difference between the outputs produced by $D_i$ and $D_{i'}$ on $u$.
    \end{claim}
    Finally, our definition of the transition function ensures that each fragment of output
    produced by $\calD$ while processing $u$ can be traced back to a fragment of output
    produced by one of the transducers $\calD_i$ while reading some prefix of $u$.
    This leads to the following claim~:
 
    \begin{claim}\label{claim:output}
        For all $u \in \dom(\calD)$, for all $1 \leq i \leq |\calD(u)|$,
        there exists a run $\rho$ of $T$ on some prefix of $u$
        such that the $i$th letter of $u$ is equal to the $i$th letter of the word produced by $\rho$.
    \end{claim}

    \subparagraph*{Bounding the state space}
    As explained in the definition of the state space of $\calD$,
    we initially define it as an infinite set of states, and we now show that
    only a finite subset is reachable from the initial state.
    To that end, we prove that every state $p =  ((p_1,w_1), (p_2,w_2), \ldots, (p_k,w_k))$
    reachable from the initial state of $\calD$ satisfies $|w_i| < N_\calT$ for every $1 \leq i \leq k$, where $N_\calT$ is the constant defined in \Cref{claim:HP}.
    This implies that the trimmed version of $\calD$ contains at most $(|Q_i| + 1) \cdot |A|^{N_\calT+1}$ states.
    Let $u$ denote a word so that $\calD$ reaches $p$ from the initial state while reading $u$.
    First, remark that by construction every pair $(p_i,w_i)$ such that $p_i = d$
    satisfies $w_i = \epsilon$, thus $|w_i| = 0 < N_\calT$.
    Now, observe that the definition of the transition function implies that
    if there is at least one index $1 \leq j \leq k$ satisfying $q_j \neq d$,
    then there exists such an index such that we also have $w_j = \epsilon$.
    Now for every other index $1 \leq i \leq k$ satisfying $q_i \neq d$,
    we get that $|w_i| = |w_i| - |\epsilon| = |w_i| - |w_j|$,
    which is equal to the difference between the length of the outputs produced by $D_i$ and $D_j$
    on $u$  by Claim~\ref{claim:words}.
    Finally,~\Cref{claim:HP} implies that this difference is smaller than $N_\calT$.

    \subparagraph*{Bounding the Hamming distance}
    We now proceed with the proof that $d(\calD,\calT) < \infty$ w.r.t.~the Hamming distance.
    Specifically, we show that every word $u \in \dom(\calT)$ satisfies
    $\dHam{\calD(u)}{\calT(u)} \leq k \cdot N_\calT$,
    where $N_\calT$ is the constant from~\Cref{claim:HP}.
    
    First, by Claim~\ref{claim:states} and the definition of the final states of $D$,
    both $\calD$ and $\calT$ have the same domain.
    Fix an input word $u \in \dom(\calT)$.
    Remark that \Cref{claim:words} and the definition
    of the output function $\lambda$ imply that
    $|\calT(u)| = |\calD(u)|$.
    Therefore, in order to bound $\dHam{\calD(u)}{\calT(u)}$
    it is sufficient to bound the number of mismatches
    $\textsf{mismatch}(\calD(u),\calT(u)) \in \mathbb{N}$ between $\calD(u)$ and $\calT(u)$.
    For each $1 \leq i \leq k$
    let $u_i$ denote the longest prefix of $u$
    for which $\calD_i$ has an initial run,
    and let $v_i$ denote the output produced by this run.
    Since $\dom(\calD) = \dom(\calT)$ there exists at least one index $j \in \{1,2,\ldots,k\}$
    such that $u$ is in the domain of $\calD_j$, which we now fix for the remainder of the proof.
    By our notation, we have $u_j = u$ and $v_j = \calD(u)$.

    By Claim~\ref{claim:output}, for every $1 \leq m \leq |v|$
    there exists $1 \leq i \leq k$ such that $v[m] = v_i[m]$.
    Consequently, every position where $v_j$ and $v$ differ
    is also a position where $v_j$ and $v_i$ differ for some $i$.
    This yields the desired result~: 
    \[
        \dHam{\calD(u)}{\calT(u)}
        = \textsf{mismatch}(v_j,v) \leq \sum_{i=1}^k \textsf{mismatch}(v_j,v_i) \leq \sum_{i=1}^k N_\calT \leq k \cdot N_\calT.\qedhere
    \]
\end{proof}
\begin{lemma} \label{lemma:seriesseqapproxHamming}
    Let $\calU = \calD_1 \calD_2 \cdots \calD_k$ be an unambiguous transducer,
    where $\calD_i$ is a sequential trim transducer for $i \in [k], k>1$.
    The series-sequential function given by the transducer $\calU$
    is approximately determinisable w.r.t.~Hamming distance
    if and only if $\calU$ satisfies the Hamming twinning property.
\end{lemma}

\begin{proof} ($\rightarrow$) By \Cref{proposition:approxtoHTPSTP}.\\
    ($\leftarrow$) Let $\calU = (Q,s,\delta,F,\lambda)$, and for every $1 \leq i \leq k$ let $\calD_i = (Q_i,s_i,\delta_i,F_i,\lambda_i)$.
    We construct a sequential transducer $\calD = (Q',G_S,\delta',F',\lambda')$ that approximately determinises $\mathcal{U}$
    by following a similar approach to the one used in the proof of Lemma~\ref{lemma:multiseqapproxHamming}.
    Specifically, we design $\calD$ to track the set of states in which $\mathcal{U}$ could currently be,
    while also keeping track of synchronised output suffixes.
    At each step, $\calD$ produces a part of these suffixes, ensuring that the state space remains finite.
    The key difference compared to the proof of Lemma~\ref{lemma:multiseqapproxHamming} is that,
    whereas in the previous case we could pick any suffix of minimal length to be produced,
    here we must be more careful to ensure that the Hamming distance remains bounded.
    We show that selecting suffixes corresponding to runs that end in the component $\calD_i$
    with the smallest index $i$ guarantees the desired outcome.

    \subparagraph*{Formal Construction}
    We begin by selecting a total ordering $\prec$ of the states of $\calU$
    that respects reachability constraints:
    \[
        \text{For all } q_i \in Q_i \text{ and } q_j \in Q_j, \text{ if } i < j \text{ then } q_i \prec q_j.
    \]
    This order is naturally extended to pairs in $Q \times B^*$,
    where pairs are compared first by the state component according to $\prec$
    and then lexicographically by the word component if states are identical.
    
    \begin{enumerate}
        \item The set of states $Q'$ is the power set $\mathcal{P}(Q \times B^*)$.
        While this definition allows arbitrarily large words,
        we will show that only a finite subset of $Q'$ is reachable from the initial state.
        \item The initial state is $G_s = \{(s_1,\epsilon)\}$.
        \item The transition function operates in two steps. Given $(G,a) \in Q' \times A$, define:
        \[
            \overline{H} = \{(q,uw) \mid \text{there exists $(p,u) \in G$ such that } \delta(p,a) = (q,w)\}.
        \]
        Then, let $C \in \mathbb{N}$ be the minimum length of the set
        $\{|w| \mid (q,w) \in \overline{H}\}$, and let
        \[
            H = \{(q,w_2) \mid (q,w_1w_2) \in \overline{H} \text{ for some } w_1 \in B^* \textup{ satisfying } |w_1| = C\}.
        \]
        We set $\delta'(G,a) = (H,v)$, where $v$ is the $C$-length prefix of the word $w_{\min}$
        occurring in the minimal pair $(q_{\min},w_{\min})$ in $\overline{H}$ according to $\prec$.
        \item The set of final states is 
        $F' = \{G \in Q' \mid \exists (p,u) \in G \text{ such that } p \in F \}.$
        \item The output function maps each state $G \in Q'$ to
        \[
            \lambda'(G) =
            \begin{cases}
                u \lambda(p) & \text{if } p \in F, \text{ where } (p,u) 
                \text{ is the minimal pair in } G \text{ such that } p \in F; \\
                \epsilon & \text{otherwise}.
            \end{cases}
        \]
    \end{enumerate}

   \subparagraph*{Key Properties}
    We establish three fundamental properties linking, for each input word $u \in A^*$,
    the initial run of $\calD$ on $u$ with the runs of $\calU$ on $u$.
    Let us fix an initial run of $\calD$~:
    \[
        \rho : s' \xrightarrow{u \mid v} G.
    \]
    First, observe that disregarding the output words in the states of $\calD$
    results in a subset construction keeping track of the
    reachable states of $\calU$.
    This leads to the following claim~:
    
    \begin{claim}\label{claim:states2}
        For all $p \in Q$,
        $\calU$ has an initial run on $u$ ending in $p$
        if and only if $(p,w) \in G$ for some $w \in B^*$.
    \end{claim}
    Next, our construction guarantees that the output suffixes stored in $G$
    remain synchronized~:
    
    \begin{claim}\label{claim:words2}
        For all $(p,u), (p',u') \in G$,
        the length difference $|u| - |u'|$
        is equal to the length difference
        between the outputs produced by the initial runs of $\calU$ on $u$ ending in $p$ and $p'$.
    \end{claim}
    Finally, our definition of the transition function ensures that the output 
    produced by $\calD$ while processing $u$ can be traced back to output
    produced by $\calU$ while reading some prefix of $u$.
    More precisely, $\delta'$ always produces output
    coming from a run of $\calU$ that ends in the smallest state with respect to $\prec$.
    To reflect this, we say that an initial run of $\calU$ is \emph{optimal}
    if it reaches the minimal state reachable on its input, and we get~:
    \begin{claim}\label{claim:output2}
        For every $1 \leq i \leq v$,
        the $i$th letter of $v$ is equal to the $i$th letter of 
        an optimal run of $\calU$ on some prefix of $u$.
    \end{claim}
    While $\calU$ might have arbitrarily many optimal runs on prefixes of $u$,
    we now build a set of $|Q|$ runs such that every optimal run on a prefix of $u$
    is a prefix of one of these runs.
    This will be crucial to bound the Hamming distance between $\calD$
    and $\calU$.
    
    Let $u = u_1u_2 \ldots u_k$ denote the decomposition of $u$ into 
    (possibly empty) subwords such that for every $1 \leq i \leq k$,
    $u_1u_2 \ldots u_i$ is the minimal prefix of $u$
    for which there exists no run of $\calU$ starting in the initial state
    and ending in $\calD_i$.
    For every $1 \leq i \leq k$, 
    let $\bar{Q}_i \subseteq Q_i$
    denote the set containing all the states $q \in Q_i$
    reachable by reading $u_1 u_2 \dots u_{i-1}$
    from the initial state.
    Moreover, let $\bar{Q} = \bigcup_{i=1}^k \bar{Q}_i$.
    For every $q \in \bar{Q}$ we define a specific run $\rho_q$ as the
    concatenation of two runs $\rho_q'\rho_q''$, where~:
    \begin{itemize}
    \item $\rho_q'$ is the initial run of $\calU$
    on $u_1 u_2 \dots u_{i-1}$ ending in $q$.
    Note that there is only one such run as $\calU$
    is an unambiguous trim transducer.
    \item $\rho_q''$ is the run of $\calD_i$ starting from $q$
    that processes the longest prefix of $u_i$ (while remaining in $\calD_i$).
    Note that there is only one such run as $\calD_i$ is sequential.
    \end{itemize}
    We now argue the following key claim:
    
    \begin{claim}\label{claim:fin}
    Every optimal initial run $\rho$ of $\calU$ on a prefix of $u$
    is a prefix of $\rho_q$ for some  $q \in Q$.
    \end{claim}

    \begin{proof}
    Let $\rho$ be an optimal run on some prefix of $u$,
    and let $1 \leq i \leq k$ be the index of the component $\calD_i$ in which $\rho$ terminates.
    By definition of the decomposition $u = u_1 u_2 \dots u_k$ the run $\rho$ terminates while processing $u_i$,
    thus its input must be of the form $u_1 u_2 \dots u_{i-1} u_i'$
    for some prefix $u_i'$ of $u_i$.
    Moreover, $\rho$ must have already entered the component $\calD_i$
    after reading $u_1 u_2 \dots u_{i-1}$.
    We decompose $\rho$ into two parts $\rho'\rho''$ defined as follows:
    \[
        \rho: p_0 \xrightarrow{u_1 u_2 \dots u_{i-1} \mid v} q \xrightarrow{u_i' \mid w} q'.
    \]
    Since $T$ is an unambiguous trim transducer we have $\rho' = \rho_q'$,
    and as $D_i$ is sequential $\rho''$ is a prefix of $\rho_q''$, 
    which implies that $\rho$ is a prefix of $\rho_q' \rho_q''$.
    This establishes the claim.
    \end{proof}

   \subparagraph*{Bounding the state space}
    The state space $Q'$ in the formal definition is infinite.
    We now show that the set of states reachable from the initial state is finite.
    More precisely, we let $N_\calU \in \mathbb{N}$ be the bound given by Lemma~\ref{claim:HP},
    and show that every initial run
    \[
        \rho : \{(s,\epsilon)\} \xrightarrow{u|v} G =  \{(p_1,w_1), (p_2,w_2), \ldots, (p_m,w_m)\}
    \]
    satisfies $|w_i| < N_\calU$ for every $1 \leq i \leq m$.
    This bounds the number of reachable states by
    \[
    2^{(|Q|) \cdot A^{N_\calU+1}}.
    \]
    Observe that $G$ is reached by successive applications of the transition function $\delta'$
    to the initial state, thus,  by definition of $\delta'$, there is a pair $(q_j,w_j) \in G$
    satisfying $w_j = \epsilon$.
    Since $|\epsilon| = 0$,
    we get that for every $(q_i,w_i) \in G$,
    the length of $|w_i|$ is equal to the difference $|w_i| - |w_j|$.
    In turn, we can apply Claim~\ref{claim:words2} to get that the length of $|w_i|$ is equal to 
    the length difference between the outputs of the initial runs of $\calU$ on $u$
    ending in $q_i$ and $q_j$.
    Finally, Lemma~\ref{claim:HP}
    guarantees that this difference is bounded by $N_\calU$
    as $\calU$ satisfies the \HTP{}.
    
    \subparagraph*{Bounding the Hamming distance}
    We now establish that for every word $u \in \dom(\calU)$,
    the Hamming distance between $\calD(u)$ and $\calU(u)$ is bounded by $|Q| \cdot N_\calU$,
    where  $N_\calU$ is the constant from Lemma~\ref{claim:HP}.
    This ensures that $\dHam{\calD}{\calU} < \infty$.
    
    By Claim~\ref{claim:states2} and the definition of the final states of $\calD$,
    we know that $\calD$ and $\calU$ have the same domain.
    Fix an input word $u \in \dom(\calU)$.
    By \Cref{claim:words2} and the definition of the output function $\lambda'$
    we get that $|\calD(u)| = |\calU(u)|$,
    therefore the distance $\dHam{\calD(u)}{\calU(u)}$
    is equal to the number of mismatches
    between $\calD(u)$ and $\calU(u)$, denoted $\textsf{mismatch}(\calD(u),\calU(u))$.
    
    Combining Claim~\ref{claim:output2} and Claim~\ref{claim:fin}
    yields that for every $1 \leq m \leq |\calD(u)|$,
    the $m$th letter of $\calD(u)$ matches
    the $m$th letter of the output $v_q$ of one of the runs $\rho_q$ with $q \in Q$.
    Consequently, every position where $\calD(u)$ and $\calU(u)$
    differ is also a position where $v_q$ and $\calU(u)$ differ for some $q \in Q$.
    This yields the desired bound through the use of \Cref{claim:HP}~: 
    \[
        \dHam{D(u)}{T(u)} = \textsf{mismatch}(\calD(u),\calU(u))
        \leq \sum_{q \in Q} \textsf{mismatch}(v_q,\calU(u))
        \leq \sum_{q \in Q} N_\calU = |Q| \cdot N_\calU. \qedhere
    \]
\end{proof}

\begin{lemma}\label{theorem:approxdetHamming}
    A rational function given by a transducer $\calT$
    is approximately determinisable w.r.t.~the Hamming distance
    if and only if $\calT$ satisfies
    both the Hamming and the strongly connected twinning property.
\end{lemma}
\begin{proof}
    ($\rightarrow$) By \Cref{proposition:approxtoHTPSTP}.\\
    ($\leftarrow$):
    First, we assume that the rational function is given by an unambiguous transducer $\calT$ and
    decompose $\calT$ into an equivalent finite union of transducers
    $\calU = \bigcup_{\pi \in \Pi} \calT_\pi$, exactly as in the proof of \Cref{proposition:approxtoATPSTP}. 
    Then  \Cref{lem:preserveTPs} guarantees that $\calU$
    satisfies both \HTP{} as well as \STP{}.

    Since $\calU$ satisfies \STP{}, for every $\pi \in \Pi$
    each SCC of $\calT_\pi$ satisfies twinning property
    if we set as initial state the unique entry point of the SCC.
    Hence we can determinise each SCC in $\calT_\pi$,
    thus obtaining a series-sequential transducer $\calT_\pi'$
    equivalent to $\calT_\pi$ that still satisfies \HTP{} by \Cref{lem:preserveTPs}.
    From \Cref{lemma:seriesseqapproxHamming},
    we can then turn each $\calT_\pi'$ into a  sequential transducer $\calD_\pi$
    for which the distance $\dHam{\calD_\pi}{\calT_\pi'} = \dHam{\calD_\pi}{\calT_\pi}$ is finite. 
    
    We then set $\dHam{\calD_\pi}{\calT_\pi} = k_\pi$ for some $k_\pi \in \mathbb{N}$,
    which yields that the distance between $\calT$ and the new transducer
    $\calU' = \bigcup_{\pi\in \Pi} \calD_\pi$ is bounded~:
    \[
    \dHam{\calU'}{\calT} = \dHam{\calU'}{\calU} = \max \{ k_\pi \mid \pi\in \Pi\} < \infty.
    \]

    Then, as $\calT$ satisfies \HTP{} so does $\calU'$ by \Cref{lem:preserveTPs},
    and as $\calU'$ is multi-sequential,~\Cref{lemma:multiseqapproxHamming}
    allows to construct a sequential transducer $\calD$ satisfying
    $\dHam{\calD}{\calU'} < \infty$.

    Now since both $\dHam{\calD}{\calU'}$ and $\dHam{\calU'}{\calT}$
    are bounded, we get that $\dHam{\calD}{\calT}$
    is also bounded by the triangle inequality.
    This shows that, as required,
    $\calT$ is approximately deteminisable w.r.t~Hamming distance.
\end{proof}

\section{Approximate decision problems for rational relations}

In this section, we consider two possible generalisations of the approximate determinisation problem to rational relations. We describe those generalisations informally. The first one asks to decide whether a rational relation is close to some rational function. We call it the approximate functionality problem. The second one, that we still call determinisation problem, amounts to decide, given a rational relation $R$, whether it is \emph{almost} a sequential function. The third generalisation we consider is an approximate uniformisation problem, which asks, given $R$, whether there exists a sequential function $s$ which is close to a function $f$, whose graph is included in $R$. We however show that this problem is undecidable. 

We now proceed with the formal definitions and statements of our results. First, we need to extend the notion of distance from functions of words to binary relations of words. 
Towards this, we use Hausdorff distance between languages, defined as \[
d_H(L,L') = \max \left\{ \sup_{w \in L} \inf_{w' \in L'} d(w,w'), \sup_{w' \in L'} \inf_{w \in L} d(w,w') \right\}.
\]
Given a metric $d$ on words, and two relations $R_1,R_2\subseteq A^*\times B^*$, the distance between $R_1$ and $R_2$ is defined as follows.
$$d(R_1,R_2) = \begin{cases} \sup \left \{\,  d_H(R_1(w), R_2(w)) \,\mid\,  w \in \dom(R_1) \right \} & \text{ if $\dom(R_1) = \dom(R_2)$} \\
 \infty & \text{ otherwise } 
 \end{cases}$$
Therefore, $d(R_1,R_2)<\infty$ iff there exists $k\in\mathbb{N}$ such that for all word $u$ in the domain and any output $v_1$ of $R_1$ on $u$, there exists some output $v_2$ of $R_2$ on $u$, such that $d(v_1,v_2)<k$, and symmetrically. In fact, $d$ is a metric on relations as shown below. 
    
\begin{proposition}\label{prop:metricrelations}
    $d$ is a metric on relations.
\end{proposition}
\begin{proof}
    When $d$ is a metric on words, the Hausdorff distance $d_H$ between languages is also a metric. Therfore, $d(R_1,R_2) = 0 \iff R_1 = R_2$ and $d(R_1,R_2) = d(R_2,R_1)$ for any two relations $R_1$ and $R_2$. It remains to show that $d(R_1,R_2) \leq d(R_1,R_3) + d(R_3,R_2)$ for relations $R_1,R_2$ and $R_3$.
  
Assume the domains of $R_1$ and $R_2$ are different. Then, either $\dom(R_1) \neq \dom(R_3)$ or $\dom(R_3) \neq \dom(R_2)$. In both cases, $d(R_1
,R_2)$ and 
$d(R_1,R_3) + d(R_3,R_2)$ are $\infty$. Therefore, assume that the domains of $R_1,R_2$ and $R_3$ are the same, call it $L$. Since for each word $w$ in $L$, $d_H(R_1(w),R_2(w)) \leq d_H(R_1(w),R_3(w)) + d_H(R_3(w),R_2(w))$ by virtue of $d_H$ being a metric, it follows that 
\setlength{\arraycolsep}{0.0em}
\begin{eqnarray*}
d(R_1,R_2) && =\sup\left\{\,d_H(R_1(w),R_2(w)) \,\mid\, w \in L\,\right\} \\
               && \leq \sup\left\{\,d_H(R_1(w),R_3(w)) + d_H(R_3(w),R_2(w))  \,\mid\, w \in L\,\right\}\\
                                          &&\leq \sup\left\{d_H(R_1(w),R_3(w)) \,\mid\, w \in L \,\right\} + \sup \left\{\, d_H(R_3(w),R_2(w))  \,\mid\, w \in L \,\right\} \\
                                          &&=  d(R_1,R_3) + d(R_3,R_2)~. \hspace{9cm}\qedhere
\end{eqnarray*}\setlength{\arraycolsep}{5pt}
\end{proof}

\subparagraph*{Approximate functionality problem}

\begin{definition}[Approximate Functionalisation]\label{def:approxfunc}
    A rational relation $R$ is approximately functionalisable w.r.t.~a metric $d$ if there exists a rational function $f$  such that $d(R,f)<\infty$. 
\end{definition}
The approximate functionality problem asks, given $R$ represented by a transducer, whether it is approximately functionalisable w.r.t. $d$. Towards this, we define the following value for a relation $R$ and metric $d$, which measures how different are output words over the same input. More precisely, it is the maximal distance between any two output words over the same input word, by $R$:
\\$
\textcolor{white}{x}~~~~~~~~~~~~~~~~~~~~~~~~~~~~~~{\sf diff}_d(R) = \sup_{u\in dom(R)} \sup_{v_1,v_2\in R(u)} d(v_1,v_2)
$

\begin{lemma}\label{lem:diffcomp}
  For a rational relation $R$ given as a rational transducer, ${\sf diff}_d(R)$ is computable for all metrics given in \Cref{table:editdistance}.
\end{lemma}
\begin{proof}
Given a rational relation $R$, we can construct a rational relation $R_o$ that consists of all pairs of output words of $R$ on any input, i.e., $R_o = \{(v_1,v_2) \mid \exists u \in \dom(R),(u,v_1),(u,v_2)\in R\}$. If $R$ is a relation defined by a transducer $\calT$, then the transducer obtained by $\calT \times \calT$ (cartesian product of $\calT$ by itself) by ignoring the input word is a transducer that defines the relation $R_o$. Observe that ${\sf diff}_d(R)$ is equivalent to the diameter of $R_o$ w.r.t.~$d$. It is shown in \cite{editdistance} that diameter of a rational relation is computable for all metrics given in \Cref{table:editdistance}. Hence, for those metrics, ${\sf diff}_d(R)$ is computable.
\end{proof}

We now characterise rational relations which are approximately functionalisable.

\begin{lemma}\label{lemma:approxfun}
   A rational relation $R$ is approximately functionalisable w.r.t.~a metric $d$ if and only if ${\sf diff}_d(R) < \infty$.
\end{lemma}

\begin{proof}

    ($\rightarrow$) Assume that $R$ is approximately funtionalisable, i.e., there exists a rational function $f$ such that $d(R,f) < \infty$. Therefore, $\dom(R) = \dom(f)$ and there exists an integer $k$ such that for each input word $u \in \dom(R)$, for each output word $v \in R(u)$, $d(v, f(u)) \leq k$ (since $f$ is a function). By triangle inequality of metric $d$, the distance between any two arbitrary output words $v_1,v_2 \in R(u)$ on any input $u \in \dom(R)$  is $d(v_1,v_2) \leq d(v_1,f(u)) + d(f(u),v_2) \leq 2k$.  Since this holds for any input in the domain of $R$, we get ${\sf diff}_d(R) = \sup_{u\in dom(R)} \sup_{v_1,v_2\in R(u)} d(v_1,v_2)  \leq 2k$.

    ($\leftarrow$) Assume that ${\sf diff}_d(R) < \infty$. Let ${\sf diff}_d(R) \leq k$ for some $k \in \mathbb{N}$. We show that any uniformiser of the relation (a function of same domain as the relation whose graph is included in the relation) is a function that approximately functionalises the relation. 
    Let $f$ be a uniformiser of $R$. We prove that $d(R,f) < \infty$. Since $f$ is a uniformiser of $R$, $\dom(R) = \dom(f)$, and for all $u\in\dom(R)$, it holds that $(u,f(u))\in R$.  Since ${\sf diff}_d(R) \leq k$, the distance between the outputs of $R$ on any input is less than or equal to $k$. Thus, for any input word $u \in \dom(R)$, for each output word $v \in R(u)$, $d(v, f(u)) \leq k$ (since $f(u)$ is also an output of $R(u)$). Hence, $d(R,f) < \infty$, i.e., $R$ is approximately functionalisable w.r.t.~$d$.
\end{proof}

Since ${\sf diff}_d(R)$ is computable (see \Cref{lem:diffcomp}), we get the following result.
\begin{theorem}
    The approximate functionality problem for rational relations given as rational transducers w.r.t.~a metric given in \Cref{table:editdistance} is decidable.
\end{theorem}

\subparagraph*{Approximate determinisation problem}
A rational relation $R$ is said to be approximately determinisable for a metric $d$ if it is almost a sequential function with respect to $d$. Formally, it means that there exists a sequential function $f$ such that $d(R,f)<\infty$. We show that the associated decision problem, that we call approximate determinisation problem, is decidable for Levenshtein family of distances. In fact, the characterisation for approximate determinisation of rational functions also holds for rational relations. 
\begin{lemma}
    A rational relation defined by a trim transducer $\calT$ is approximate determinisable w.r.t.~Levenshtein family ($d_l,d_{lcs},d_{dl}$) if and only if $\calT$ satisfies \ATP{} and \STP{}.
\end{lemma}

\begin{proof}
    Let $R$ be a rational relation given by a transducer $\calT$, and let $d \in \{d_l,d_{lcs},d_{dl}\}$. The proof of direction $(\rightarrow$) is the same as \Cref{proposition:approxtoATPSTP} when the function is replaced with a relation.
    For the other direction, assume that $\calT$ satisfies both \ATP{} and \STP{}. We first show that if $R$ is approximately functionalisable then it is approximately determinisable when $\calT$ (that defines $R$) satisfies both \ATP{} and \STP{}. Let $f$ be a rational function that approximately functionalises $R$ w.r.t.~Levenshtein distance, i.e., $d(R,f) < \infty$. Thus, an unambiguous transducer $\calF$ that defines $f$ satisfies $d(\calT,\calF) < \infty$. From \Cref{lem:preserveTPs}, because $\calT$ satisfies both \ATP{} and \STP{} and $d(\calT,\calF) < \infty$, it follows that $\calF$ also satisfies \ATP{} and \STP{}. Using \Cref{theorem:approxdetlev}, the rational function $f$ is approximately determinisable. Let $g$ be a sequential function that approximately determinises $f$ w.r.t~Levenshtein, i.e, $d(f,g) < \infty$. Since $d$ is a metric on relations (see \Cref{prop:metricrelations}), $d(R,g) \leq d(R,f) + d(f,g)$. Since both $d(R,f)$ and $d(f,g)$ are finite, we get $d(R,g) < \infty$. Hence, the rational relation $R$ is approximately determinisable.

    Now it suffices to show that $R$ is approximately functionalisable w.r.t.~Levenshtein family of distance. Towards this, we prove that ${\sf diff}_{d}(R) < \infty$ when $\calT$ satisfies \ATP{}. Observe that ${\sf diff}_{d}(R) = \dia{R_o}{d}$ where $R_o = \{(v_1,v_2) \mid \exists u \in \dom(R),(u,v_1),(u,v_2)\in R\}$ is a rational relation that consists of all pairs of output words of $R$ on any input. Using \Cref{lem:ATPtobound}, $\dia{R_o}{d} \leq N_\calT$ when $\calT$ satisfies \ATP{}. Hence,  ${\sf diff}_{d}(R) = \dia{R_o}{d} < \infty$, and $R$ is approximately functionalisable w.r.t.~$d$ by virtue of \Cref{lemma:approxfun}. \qedhere
    
\end{proof}
Since \ATP{} and \STP{} are decidable for transducers (\Cref{lem:decideTPs}), we obtain the following.
\begin{theorem}
    The approximate determinisation problem for rational relations given as rational transducers w.r.t.~a metric $d \in \{d_l,d_{lcs},d_{dl}\}$ is decidable.
\end{theorem}
\subparagraph{Approximate uniformisation problem}

Given a relation $R\subseteq A^*\times B^*$, a \emph{uniformiser} of $R$ is a function $U : A^*\rightarrow B^*$ such that $\dom(R) = \dom(U)$ and 
for all $u\in\dom(R)$, it holds that $(u,U(u))\in R$. It is known that any rational relation admits a rational uniformiser~\cite{KOBAYASHI1969}, which is not true if we seek for a sequential uniformiser. Moreover, the problem of deciding whether a given rational relation admits a sequential uniformiser is undecidable~\cite{LoC}. We consider an approximate variant of this problem.

\begin{definition}[Approximate Uniformisation]\label{def:approxuni}
   A rational relation $R$ is approximately uniformisable w.r.t.~a metric $d$ if there exists a 
    uniformiser $U$ of $R$ and a sequential function $f$ such that $d(U,f)<\infty$. In that case, we say that $R$ is $d$-approximate uniformisable by a sequential function. 
\end{definition}
In other words, $R$ is $d$-approximate uniformisable by a sequential function $f$ iff there exists an integer $k \in \mathbb{N}$ such that for all $u\in dom(R)$, there exists $(u,v)\in R$ with $d(v,f(u))\leq k$. 

\begin{theorem}\label{thm:app-un}
Checking whether a rational relation is $d$-approximate uniformisable for $d \in \{d_l,d_{lcs},d_{dl}\}$ is undecidable.
\end{theorem}

\begin{proof}
Let $\phi_1 : \{1,\dots,k\}\rightarrow B^*$ and $\phi_2 :\{1,\dots,k\}\rightarrow B^*$ be two morphisms defining an instance of the post correspondence problem (PCP), which asks whether there exists $w\in\{1,\dots,k\}^+$ such that $\phi_1(w) = \phi_2(w)$. Let $A = B\uplus \{1,\dots,i,a,b,\#\}$.

Let $R$ be the relation which as input takes any word of the form $w_1\#w_2\#\dots w_m\#X$
where $w_i\in \{1,\dots,k\}^*$ for $1\leq i\leq m$ and $X\in\{a,b\}$. Let us define the outputs:
\begin{itemize}
\item If $X = a$, then the only output is $\phi_1(w_1)\#\phi_1(w_2)\#\dots \phi_1(w_m)\#$.
\item If $X = b$, then any word of the form 
$v_1\#\dots v_m\#\text{ where } v_i\neq \phi_2(w_i)$ for all $1\leq i\leq m$ is a valid output.
\end{itemize}
It can be shown that $R$ is rational, recognizable by some transducer $\calT$.
We now show that $R$ is approx-uniformisable iff PCP has no solution iff $R$ is exact-uniformisable. First, suppose that PCP has a solution $w$ and $R$ is approx-uniformisable by some
sequential transducer $\mathcal{D}$, i.e. $d(\trans,{\cal D})\leq K$
for some $K$. Consider inputs of the form $u_{\ell} = (w\#)^\ell$ for
$\ell\geq 0$, and let $\alpha_\ell,\alpha_a,\alpha_b$ be such that $q_0\xrightarrow{u_\ell|\alpha_\ell}_{\cal D} q,
q\xrightarrow{a\mid \alpha_a}_{\cal D} q_f\text{ and } q\xrightarrow{b\mid
  \alpha_b}_{\cal D} p_f$, 
where $q_0$ is the initial state, $q$ is a state and $q_f,p_f$ are
final states of ${\cal D}$. Since $d(\trans,{\cal D})\leq K$, for all $\ell$,
$d((\Phi_1(w)\#)^\ell, \alpha_\ell\alpha_a)\leq K$ holds.
Similarly, for all $\ell$, there exist $v_1,v_2,\dots,v_\ell$ all different from $\Phi_2(w)$ such that
$$d(v_1\#v_2\#\dots v_\ell\#, \alpha_\ell\alpha_b)\leq K.$$
Since $d(\alpha_\ell\alpha_a,\alpha_\ell\alpha_b)$ is uniformly bounded for all $\ell$ by some $M$, by applying triangular
inequalities, we get that for all $\ell$, there exist $v_1,\dots,v_\ell$ all different from $\Phi_2(w)$ such that
$$
d((\Phi_1(w)\#)^{\ell}, v_1\#v_2\#\dots v_\ell\#)\leq 2K+M
$$
Take $\ell = 4K+2M+2$, and fix a sequence of at most $2K+M$ edits from
\\$(\Phi_1(w)\#)^{\ell} = (\Phi_1(w)\#\Phi_1(w)\#)^{2K+M+1}$ to
$v_1\#v_2\#\dots v_\ell\#$. In this sequence, at least one copy of
$(\Phi_1(w)\#\Phi_1(w)\#)$ is not edited by the sequence. Therefore, 
there exists some $i$ such that $\#v_i\# = \#\Phi_1(w)\#$ and hence $v_i = \Phi_1(w)$. It is a contradiction since $v_i\neq \Phi_2(w) = \Phi_1(w)$. Therefore $R$ is not approximately uniformisable.

Conversely, if there is no solution to PCP, then the following
sequential function is an exact uniformiser of $R$: on any input of
the form $w_1\#\dots \# w_mX$ it outputs $\Phi_1(w_1)\#\dots \#\Phi_1(w_m)$.\qedhere
\end{proof}

\section{Future works}
In this paper, we proved that approximate determinisation is decidable for functional transducers, by checking various twinning properties. We have shown that \HTP{} and \STP{} are decidable in \textsc{PTime}, which entails that approximate determinization of functional transducers for the Hamming distance is decidable in \textsc{PTime}. For the other distances, such as Levenshtein distance, the time complexity is doubly exponential, as deciding conjugacy of a rational relation, and hence \ATP, is doubly exponential~\cite{decidingconjugacy}.
We conjecture that this is suboptimal,
and leave as future work finding a better upper-bound. 

We show that approximate uniformisation is undecidable for Levenshtein family. We leave the case of Hamming distance as future work. The current undecidability proof for Levenshtein family, based on PCP, heavily requires that the lengths of the output words produced on transitions can differ, which may not guarantee that the total output lengths are the same, which is necessary to have a finite Hamming distance. Our proof also does not extend to the
letter-to-letter setting, where both the given transducer
and the required uniformiser process \emph{and produce} a single letter
on every transition.
This problem is closely related to the standard Church synthesis for regular specifications,
with the modification that the strategy to be synthesized is allowed to make
a bounded number of errors.
To our knowledge, this variant has not been studied in the literature.

Finally, we studied approximate decision problems up to finite distance. Another interesting question is to consider their  ``up to distance $k$'' variant, where $k$ is given as input.
    
\bibliography{reference} 

\end{document}